\pdfoutput=1
\documentclass[letterpaper]{article} 
\usepackage{aaai23}  
\usepackage{times}  
\usepackage{helvet}  
\usepackage{courier}  
\usepackage[hyphens]{url}  
\usepackage{graphicx} 
\usepackage{natbib}  
\usepackage{caption} 
\usepackage{amsmath,amsfonts}
\usepackage{amsthm}

\usepackage{algorithm}
\usepackage{todonotes,comment}
\usepackage[algo2e,ruled]{algorithm2e}

\usepackage{thmtools}
\usepackage{thm-restate}

\usepackage{newfloat}
\usepackage{listings}
\DeclareCaptionStyle{ruled}{labelfont=normalfont,labelsep=colon,strut=off} 
\floatstyle{ruled}
\newfloat{listing}{tb}{lst}{}
\floatname{listing}{Listing}
\title{On Manipulating Weight Predictions in Signed Weighted Networks}
\title{On Manipulating Weight Predictions in Signed Weighted Networks}
\author {
    Tomasz Lizurej,\textsuperscript{\rm 1,2}
    Tomasz Michalak, \textsuperscript{\rm 1,2}
    Stefan Dziembowski \textsuperscript{\rm 1,2}
}
\affiliations {
    \textsuperscript{\rm 1} University of Warsaw\footnote{This research was supported by NCN Grant 2019/35/B/ST6/04138 and  ERC Grant 885666.}\\
    \textsuperscript{\rm 2} IDEAS NCBR\\
    tomasz.lizurej@crypto.edu.pl, tpm@mimuw.edu.pl, stefan.dziembowski@crypto.edu.pl
}
\usepackage{bibentry}

\newtheorem{thm}{Theorem}

\newtheorem{prob}{Problem}

\newtheorem{axiom}{Axiom}

\newcommand{\VC}{\mathit{VC}}
\newcommand{\NP}{\mathit{NP}}
\newcommand{\TP}{\mathit{TP}}
\newcommand{\FGA}{\mathit{FGA}}
\newcommand{\DMT}{\mathit{DMT}}
\newcommand{\IMT}{\mathit{IMT}}
\newcommand{\DNR}{\mathit{DNR}}
\newcommand{\INR}{\mathit{INR}}
\newcommand{\dom}{\mathit{dom}}
\newcommand{\pred}{\mathit{Pred}}
\newcommand{\succe}{\mathit{Succ}}

\newcommand{\indeg}{\mathit{indeg}}

\begin{document}

\maketitle
\begin{abstract}
Adversarial social network analysis studies how graphs can be rewired or otherwise manipulated to evade social network analysis tools. While there is ample literature on manipulating simple networks, more sophisticated network types are much less understood in this respect. In this paper, we focus on the problem of evading $\FGA$---an edge weight prediction method for signed weighted networks by \cite{kumar2016edge}. Among others, this method can be used for trust prediction in reputation systems. We study the theoretical underpinnings of $\FGA$ and its computational properties in terms of manipulability. Our positive finding is that, unlike many other tools, this measure is not only difficult to manipulate optimally, but also it can be difficult to manipulate in practice.
\end{abstract}

\section{Introduction} 
Adversarial social network analysis studies how networks can be rewired or otherwise manipulated to falsify network examination. In particular, many works in this body of research studied how to manipulate classic tools of social network analysis such as centrality measures~\citep{crescenzi2016greedily,bergamini2018improving,was2020manipulability}, and community detection algorithms~\citep{Waniek2018hiding,fionda2017community,chen2019ga}. Also, a rapidly growing body of works studies adversarial learning on graphs using deep learning ~\citep{chen2020survey}.

While most of the above literature focused on simple networks, in this paper, we consider a more complex model of weighted signed networks. 
In this class of networks, links are labeled with real-valued weights representing positive or negative relations between the nodes~\citep{leskovec2010predicting,leskovec2010signed,tang2016survey}. An important application of signed weighted networks is the modelling of trust networks/reputation systems, 
the goal of which is to avoid transaction risk by providing feedback
data about the trustworthiness of a potential business partner~\citep{resnick2000reputation}. 
As an example, let us consider the cryptocurrency trading platform Bitcoin OTC~\citep{kumar2016edge}. In this platform, users are allowed to rate their business partners on the scale $\{-10,-9,\ldots,10\}$, and the ratings are publicly available in the form of a who-trusts-whom network. A 6-node fragment of this network is presented in Figure~\ref{fig:intro}.

\begin{figure}[b]
    \centering
    \includegraphics[width=0.25\textwidth]{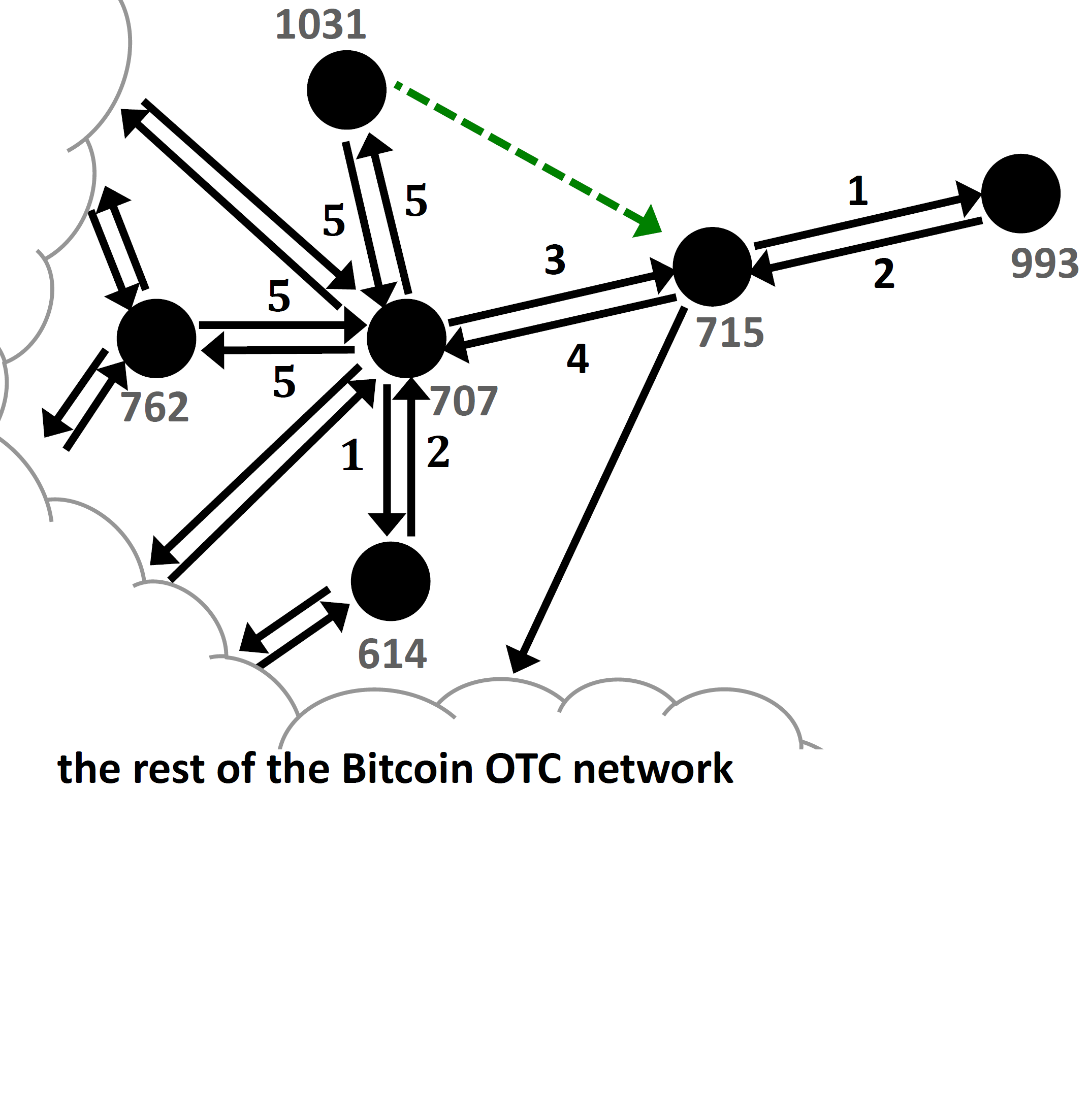}
\vspace{-1.3cm}
\caption{A fragment of the Bitcoin OTC network composed of nodes 993, 715, 707, 614, 1031, 762.}\label{fig:intro} 
\end{figure}

A user who thinks of doing a transaction with another user for the first time can use the information from such a who-trust-whom network to predict the potential risk. Technically, given a trust network modeled as a weighted signed network, predicting trust amounts to predicting the weights of potential new edges. A well-known edge weight prediction method, called $\FGA$, was proposed by \cite{kumar2016edge}. $\FGA$ is based on two measures of node behavior: the \textbf{goodness} that evaluates how much other nodes trust a given node, and the \textbf{fairness} that captures how fair this node is in rating other nodes. Both concepts have a mutually recursive definition that converges to a unique solution. Most 
importantly, Kumar et al. showed that $\FGA$ is effective in predicting edge weights, i.e., the level of trust between unlinked nodes. For example, in Figure~\ref{fig:intro}, the trust of node 1031 towards node 715 is predicted by $\FGA$ to be 2.26.


While $\FGA$ seems to be an interesting tool to apply in practice, little is known about its resilience to malicious behaviour. In this paper, we present the first study of manipulating the $\FGA$ function by a \textit{rating fraud}
~\citep{cai2016fraud,mayzlin2014promotional}. It involves fraudulent raters to strategically underrate or overrate other users for their own benefit. To magnify the strength of the manipulation, the attacker may create and act via multiple fake user identities. 
Such so called \textit{Sybil attacks} are especially tempting in environments such as cryptocurrency trading platforms where creating a new identity is affordable. Rating fraud attacks may be \textit{direct}---when targeted nodes are rated directly by the attackers---and \textit{indirect}---when the attackers try to manipulate the neighbourhood of the target nodes rather than the target nodes themselves (see Figure~2). 
It is important to distinguish between direct and indirect manipulations, as in some situations, only indirect ones will be practical. This may be the case on e-commerce platforms such as e-Bay, where nodes rate each other only after completing a transaction. When a retailer of expensive products is the target, the cost of a direct attack can be prohibitive. Hence, an indirect attack becomes an attractive alternative---it may be much cheaper to attack through the clients or business partners of such an expensive retailer (see the next section for an example).

Our contributions can be summarised as follows:
\begin{itemize}
    \item To analyze the theoretical underpinnings of the $\FGA$ measure, we propose the system of basic axioms for both fairness and goodness. We prove that together they uniquely determine the $\FGA$ measure;
    \item Next, we formulate the issue of manipulating the $\FGA$ measures of some target group of nodes as a set of computational problems. We then prove that all these problems are $\NP$-hard and $W[2]$-hard, i.e., $\FGA$ is, in general, hard to manipulate.
    \item Given the hardness of attacking a group of nodes, we then focus our analysis on targeting a single node - directly or indirectly. We first prove that direct attacks on a single node are easy, i.e., it is easy for an attacker to directly rate the target node to change the sign of  her \emph{goodness} value. As for an indirect attack, we show analitycally that for some class of networks (which we call \textit{minimum-$k$-neighbour graphs}, since we require that every node in this network has $indegree$ and $outdegree$ at least $k$), we can bound the strength of indirect attacks. Our positive finding is that, in this case, $\FGA$ measure turns out to be rather difficult to manipulate.
    \item In our experimental analysis, we first evaluate two benchmarks: (a) the strength of the aforementioned direct attack, and (b) the strength of an indirect attack based on a simple greedy approach. The latter one turns out to be very ineffective. Next, we analyse an improved greedy approach by attacking at a larger scale in every step. This approach, although costly, proves to be often effective.
\end{itemize}

\section{Preliminaries}
\noindent A Weighted Signed Network (WSN) is a directed, weighted graph $G = (V, E, W )$, where $V$  is a set of users, $E \subseteq  V\times V$ is a set of (directed) edges, and $\omega: E\rightarrow [ 1, +1]$ is a weight function that to each $(u,v) \in E$ assigns a value between $-1$ and $+1$ that represents how $u$ rates $v$. For any directed edge $(u,v) \in E$, let us denote by $\overline{(u,v)}$ the edge in the opposite direction, i.e., $\overline{(u,v)} = (v,u)$. For any set of directed edges $E$, denote by $\overline{E} = \{\overline{e}: e \in E\}$. Furthermore, let $P$ be a set of pairs of nodes of cardinality $n$, i.e, $P = \{\{u_1,v_1\},..., \{u_n, v_n\}\}$. The domain of $P$ is the set of nodes that make the pairs in $P$, i.e. 
$\dom(P)=\{u: u \in \{u,v\} \in P\}$. Finally, we write $\pred(v)$ (resp. $\succe(v)$) to denote the set of \textit{predecessors} (resp. \textit{successors}) of $v$ (resp. $u$) defined as follows: $\pred(v) = \{u: u \in (u,v) \in E\}$ (resp. $\succe(u) = \{v: v \in (u,v) \in E\}$). 

For a square matrix $M^{m \times m}$, we define $||M||_{\infty} = \max_{1 \leq i \leq m } \sum_{j=1}^m m_{ij}$, $||M||_{1} = \max_{1 \leq j \leq m } \sum_{i=1}^m m_{ij}$. It is also known that $||M \times M ||_{\infty} \leq ||M||_{\infty} \cdot ||M||_{1}$ and $||M \times M ||_{\infty} \leq ||M||_{1} \cdot ||M||_{1}$ (see \url{https://en.wikipedia.org/wiki/Matrix_norm}).

Kumar et al.~(\citeyear{kumar2016edge}) define a recursive function, $\FGA$, that assigns to each vertex of a weighted directed graph two values: \textit{fairness} and \textit{goodness}, $(f(v), g(v))$. The first one, $f(v)$, assigns a real value from range $[0,1]$ to $v$ that indicates how \emph{fair} this node is in rating other nodes. The second one, $g(v)$, assigns a value from range $[-r,r]$ to $v$ indicating how much \emph{trusted} this node is by other nodes (for simplicity we assume that $r=1$ in this paper). Finally we define an in-degree ($indeg(u)$) and out-degree ($outdeg(u)$) of a node $u \in V$. $indeg(u) = |\{(v,u): (v,u) \in E\}|$ and $outdeg(u) = |\{(u,v): (u,v) \in E\}|$.
Kumar et al.'s recursive formula for $(f(v), g(v))$ is as follows:
\begin{eqnarray}
g(v) = \frac{1}{indeg(v)}\sum_{u \in \pred(v)} f(u)\times \omega(u,v) \label{def:goodness}\\
f(u) = 1-\frac{1}{outdeg(u)}\sum_{v \in \succe(u)} \frac{|\omega(u,v)-g(v)|}{2}, \label{def:fairness} 
\end{eqnarray}
\noindent where $g(v) = 1$ for $v \in V$ with $indeg(v) = 0$, and $f(v) = 1$ for $v \in V$ with $outdeg(v)=0$. 

\begin{figure}[t]
    \centering
    \includegraphics[width=0.45\textwidth]{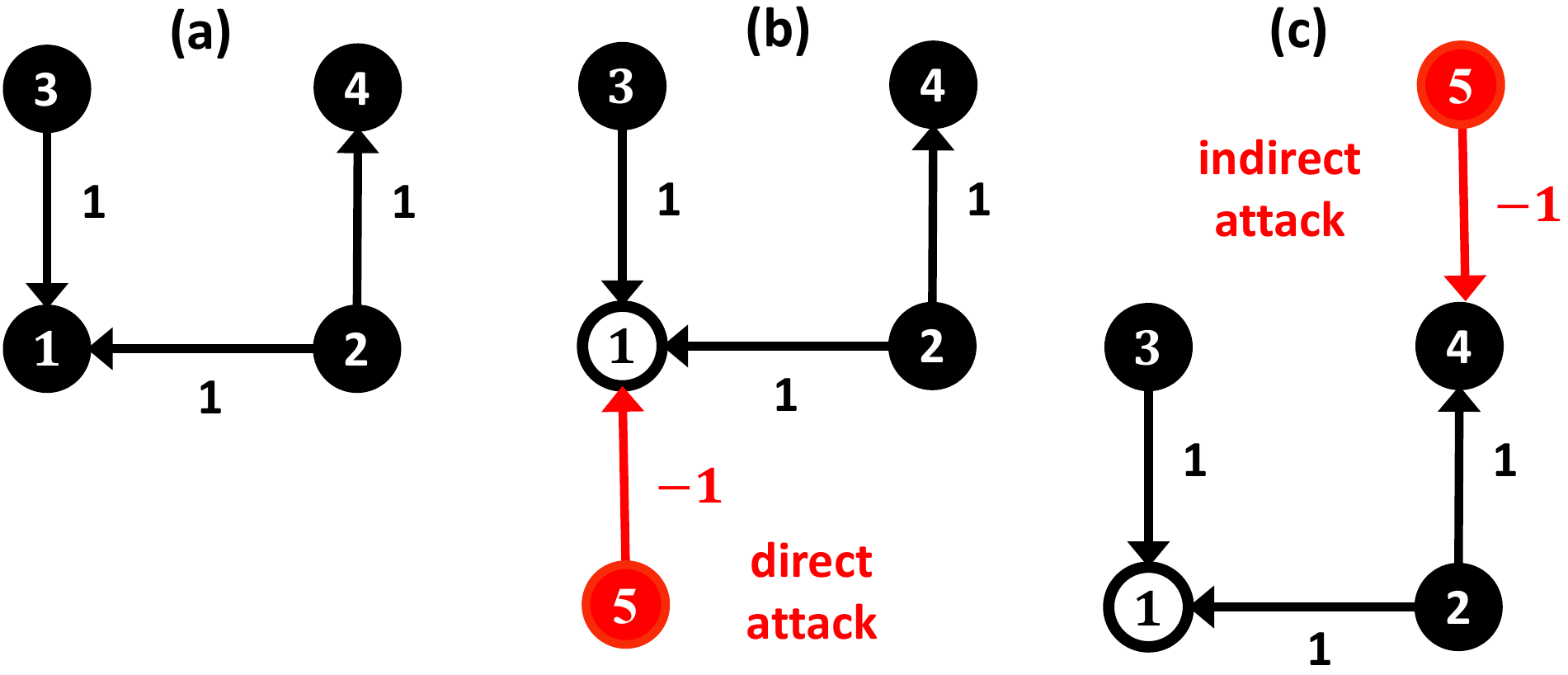}
\vspace{-0.1cm}
\begin{center}
\begin{tabular}{ccccccc}
\hline
Network &  \multicolumn{2}{c}{\textbf{(a)}} & \multicolumn{2}{c}{\textbf{(b)}} & \multicolumn{2}{c}{\textbf{(c)}}\\
\hline
node $v$   & $g(v)$ & $f(v)$  & $g(v)$ & $f(v)$ &  $g(v)$ & $f(v)$  \\ \hline
$1$ & $1$   &    $1$   &  $\mathbf{0.4}$& $1$       &  $\mathbf{0.83}$ & $1$   \\ 
$2$ & $1$   &    $1$   &  $1$& $0.8$       &  $1$  & $0.75$ \\ 
$3$ & $1$   &    $1$   &  $1$& $0.7$       &  $1$   & $0.92$ \\ 
$4$ & $1$   &    $1$   &  $0.8$& $1$       &  $0.17$   & $1$ \\ 
$5$ &    &       &  $1$& $0.3$       &  $1$   & $0.42$ \\ \hline
\end{tabular}
\end{center}
\caption{Sample networks and two types of attacks.}\label{fig:FGAexp} 
\end{figure}

Kumar et al.~(\citeyear{kumar2016edge}) showed that this function can be computed iteratively starting from $f^{(0)}(u)=g^{(0)}(u)=1$. Theorem 1 from the aforementioned work states that at each step, $t$, the estimated values $f^{(t)}(u)$,  $g^{(t)}(u)$ get closer to their limits $f^{(\infty)}(u)$,  $g^{(\infty)}(u)$,  i.e. we  have $|f^{(\infty)}(u)-f^{(t)}(u)| < \frac{1}{2^{t}}$ and $|g^{(\infty)}(u)-g^{(t)}(u)| < \frac{1}{2^{t-1}}$.
The $\FGA$ function can be used for predicting the weight of some not-yet existing (or unknown) edge $(u, v) \in V \times V \setminus E$ by computing the product: $\omega(u,v) = f(u) \times g(v)$. 

As an example of the $\FGA$ function and how it could be attacked, let us consider Figure~2. Network (a) is a benchmark, where every node rates others with the highest possible value.
In network (b), a new node 5 is used to perform a \textbf{direct attack} by rating node $1$ with the worst possible value of $-1$. This decreases the goodness of node $1$ to $0.4$. However, as argued in the introduction such a direct attack can be prohibitively costly. 
Nevertheless, given the definition of the $\FGA$ function, node $5$ can also perform an indirect attack on node $1$. This can be done, for instance, by directly attacking node $4$. As node $4$ has already been rated positively by node $2$, an opposite rating introduced by $5$ will decrease the fairness of $2$. In particular, comparing network~(c) to~(a) in Figure~\ref{fig:FGAexp}, the fairness of $2$ decreased from $1$ to $0.75$. This lower fairness means that node's $2$ ratings are less meaningful in network~(c) than in network~(a). Hence, the goodness of node $1$ decreases to $0.83$.

\section{Axiomatization}
Our first result is an axiom system that completely characterizes the $\FGA$. Below we present a comprehensive summary, while the details will are available in the appendix of the paper.

We begin with the characterization of the goodness part of the $\FGA$ function. Recall that the idea behind the goodness of $v$ is that it should reflect how this node is rated by its predecessors. Moreover, the ratings of the fairer predecessors should count more. We translate these high-level requirements into the following axioms:
\begin{itemize}
\item SMOOTH GOODNESS---let all predecessors of a particular node, $v \in V$, be unanimous in how they rate $v$ and let their fairness be the same. Now, let us assume that their fairness increases equally, i.e., intuitively, the nodes that rate $v$ become more trustworthy. Then, we require that this will result in an increase of the goodness of $v$, and that this increase is proportional to the increase of the fairness of $v$'s predecessors;
\item INCREASE WEIGHT---let the predecessors of $v$ be all equally fair and unanimous in how they rate $v$. Now, let them increase their rating of $v$ equally. Then, we require that the goodness of $v$ increases  and that this increase is proportional to the increase in how $v$ is rated;
\item MONOTONICITY FOR GOODNESS---the predecessors with higher fairness should have a bigger impact on the goodness of $v$. Similarly, higher weights should also have a bigger impact;
\item GROUPS FOR GOODNESS---let $v$ be rated by $k$ groups of the predecessors and let the nodes in each group be homogeneous and unanimous w.r.t. $v$. What is then the relationship between the impact these groups have on the goodness of $v$? 
In line with the previous axioms,
we require that the goodness of $v$ should be equal to the \emph{weighted average} of the ratings achieved when these groups separately rate $v$;
\item MAXIMAL TRUST---this basic condition requires that any if all the predecessors of $v$ have the highest possible fairness and their ratings are the highest possible, then the goodness of $v$ should be the highest possible;
\item BASELINE FOR GOODNESS---a non-rated node has the goodness of $1$.
\end{itemize}
Our first result is that the above axioms uniquely define the goodness part of the $\FGA$ function.

Let us now characterise the fairness part of the $\FGA$ function. Recall that the idea behind the fairness of $v$ is that it should reflect how the ratings given by this node agree with the ratings given by other nodes, i.e. how erroneous $v$ is. In this respect, we have the following axioms:
\begin{itemize}
    
    \item SMOOTH FAIRNESS---this axiom stipulates that the fairness of a node making an average error is an average of the fairness values of nodes making extreme errors;
    
    \item MONOTONICITY FOR FAIRNESS---our first axiom stipulates that the fairness of a node that rates more accurately than before should rise;
    \item GROUPS FOR FAIRNESS---if the nodes rated by $v$ can be divided into $k$ groups such that each node in a particular group is rated by $v$ in the same way, then the fairness of $v$ should be equal to the \emph{weighted average} of $v$'s fairness in a setting where $v$ rates these groups separately;
    \item OBVIOUS FAIRNESS METRIC---here, we stipulate that when a node makes maximal errors when rating all of its neighbors, then its fairness should be $0$, and when there is no error, then the fairness is $1$;
    \item BASELINE FOR FAIRNESS---the fairness of a node that rates noone is~1.
\end{itemize}
The above axioms uniquely define the fairness part of the $\FGA$ function. In summary, all the above axioms uniquely define the $\FGA$ function.
\begin{restatable}{thm}{axiomatization}
The SMOOTH GOODNESS, INCREASE WEIGHT, MONOTONICITY FOR GOODNESS, MAXIMAL TRUST, GROUPS FOR GOODNESS, BASELINE FOR GOODNESS axioms and the SMOOTH FAIRNESS, MONOTONICITY FOR FAIRNESS , OBVIOUS FAIRNESS METRIC, GROUPS FOR FAIRNESS, and BASELINE FOR FAIRNESS axioms uniquely define the $\FGA$ function.
\end{restatable}





\section{Complexity of attack}
Let us now study the complexity of manipulating $\FGA$.
\paragraph{Attack models}
Given $G=(V, E, \omega(E))$, let $A \subseteq V$ be a set of attackers. We define two types of the $A$'s objectives:

\begin{itemize}
    \item \textbf{targeting potential links} --- here, the target set $\TP$ is composed of disconnected pairs of nodes from $V\setminus A$:
    \begin{eqnarray}\label{eqn:TPdefinition}
    \TP \subseteq \left\{ \{u,v\}: u,v \in V  \setminus A \wedge (u,v),(v, u) \notin E\right\}.
    \end{eqnarray}
    Intuitively, the aim is to change the predicted weight of the potential links between the pairs from $\TP$.
    \item \textbf{targeting nodes} --- here, the target set is $T \subseteq V \setminus A$. Intuitively, the goal is to alter the targets' reputation.
\end{itemize}
The attackers can make the following types of moves:
\begin{itemize}
    \item \textbf{edge addition} --- the attackers can add an edge $(u,v)$ to $G$, where $u \in A$, $v \in T$, $(u,v) \notin E$, and with the weight $\omega(u,v) \in [-1,1]$. This corresponds to the attacker $u \in A$ rating node $v \in T$ for the first time.
    \item \textbf{weight update} --- an attacker $u \in A$ can update the weight of an existing edge $(u,v) \in G$ to some value $\omega(u,v) \in [-1,1]$. This corresponds to a modification of the existing rating by the attacker.
\end{itemize} 
All the attackers are allowed to make no more than $k$ such moves in total. We will refer to $k$ as a budget.\footnote{We place no constraints on how the attackers distribute this budget among themselves. In an extreme case, a single attacker can do all $k$ actions.} 

We will now formalize our computational problems. In the first one, the attackers aim at modifying the predicted weights between the pairs of nodes in $\TP$ to decrease them below (increase above) a certain threshold. This attack corresponds to breaking potential business connections.

\begin{prob}[DECREASE (INCREASE) MUTUAL TRUST, $\DMT$ ($\IMT$)]
Given a weighted signed network $G = (V, E, \omega)$,  a set of attacking nodes $A \subseteq V$, a target set of disconnected pairs of nodes $TP$ as defined in eq.~\ref{eqn:TPdefinition}, an intermediary set $I \subseteq V$, the budget $k$, and a threshold $t \in [-1,1]$, decide for all $\{u,v\} \in \TP$ whether it is possible to decrease (increase) the value of either predicted weight $f(u)\times g(v)$ or  $f(v) \times g(u)$ 
to or below (above) the threshold $t$ by making no more than $k$ edge additions or weight updates with the restriction that the attackers $u \in A$ are rating only the nodes from the intermediary set $I$.
\end{prob}

In the second problem, the attackers aim at  altering the goodness value of the nodes from a target set $T$. This attack corresponds to spoiling the reputation of the target nodes.
\begin{prob}[DECREASE (INCREASE) NODES RATING, $\DNR$ ($\INR$)]
Given WSN $G = (V, E, \omega)$,  a set of attackers $A \subseteq V$, a target set $T \subseteq V\setminus A$, an intermediary set $I \subseteq V$, the number of possible moves $k$, and threshold $t \in [-1,1]$,
decide whether it is possible, for all $v\in T$, to decrease (increase) the goodness of each vertex $v$ to or below (above) threshold $t$ by making no more than $k$ edge additions or weight updates with the restriction that the attackers $u \in A$ are rating only the nodes from the intermediary set $I$.
\end{prob}

\paragraph{Hardness Results}
We first consider $\DMT$ ($\IMT$).
\begin{restatable}{thm}{vertexcover}
\label{thm:vertexcover}
Solving the $\DMT (\IMT) = \left(G=(V, E, \omega),\right.$ $\left.A, \TP, I, t, k\right)$ problem is $\NP$-hard.
\end{restatable}

\begin{figure*}[t]
    \centering
    \includegraphics[width=0.33\textwidth]{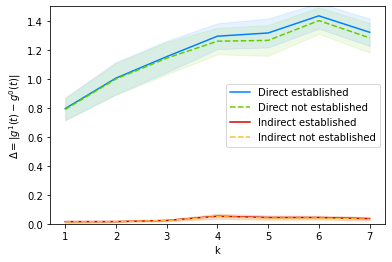}
    \includegraphics[width=0.33\textwidth]{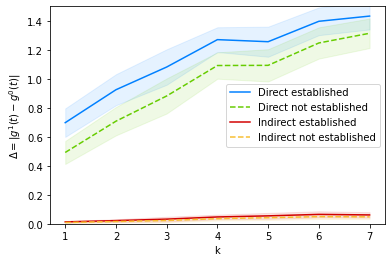}
    \includegraphics[width=0.33\textwidth]{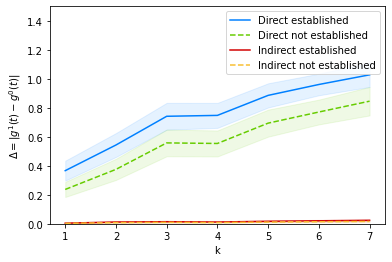}
    \caption{The comparison of direct/indirect, established/non-established attacks for Bitcoin OTC, Bitcoin Alpha and RFA Net networks.}
    \label{fig:dg}
\end{figure*}


\begin{algorithm}[t]
 \KwData{$A, T = \{t\}, G$}
 \For{$a \in A$}{
  add an edge $(a,t)$, $\omega(a,t) =-1$ to the graph $G$\;
 }
 \caption{Direct attack}
 \label{alg:direct}
\end{algorithm}
\begin{algorithm}[t]
 \KwData{$A, T = \{t\}, G$}
 sort nodes in $A$ by their fairness score\\
 \For{$a \in sorted(A)$}{
   $N_1 \leftarrow \pred(t)$\\
    find a neighbor $n_2 \in \succe(n_1) \setminus \{t\}$ of a neighbor $n_1 \in N_1$ that minimizes the goodness value of $t$, when adding  an edge $(a, n_2)$ with weight $\omega(a, n_2) = 1$ or $\omega(a, n_2) = -1$\\
    add the edge to the graph $G$\\
 }
 \caption{Indirect attack}
 \label{alg:indirect}
\end{algorithm}

    




\begin{restatable}{thm}{setcover}
\label{thm:dnr}
Solving the $\DNR (\INR) = \left(G=(V, E, \omega),\right.$ $\left.A, T, I, t, k\right)$ problem is $\NP$-hard.
\end{restatable}
Proof of the above theorems can be found in the appendix of the paper.
\paragraph{Parametrized complexity} The following results, in terms of the $W$-hierarchy for the parameterized algorithms~\cite{parameterized}, hold:
\begin{restatable}{thm}{wparam}
\label{thm:w2}
$\DNR(\INR)$ parameterized by $k$ is  $W[2]$-hard.
\end{restatable}


\begin{restatable}{thm}{wparamdmt}
\label{thm:w2dmt}
$DMT(IMT)$ parameterized by $k$ is $W[2]$-hard.
\end{restatable}
Proof of the above theorems can be found in the appendix of the paper.

\begin{figure*}[t]
    \centering
    \text{Bitcoin OTC} \\
    \includegraphics[width=0.33\textwidth]{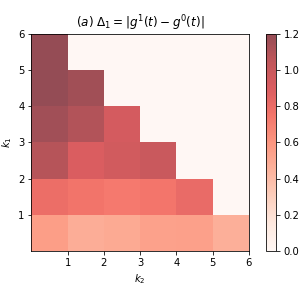}
    \includegraphics[width=0.33\textwidth]{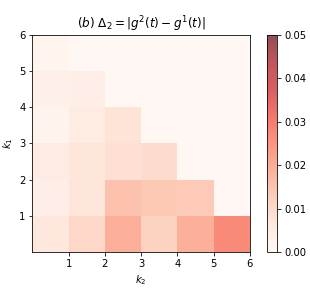}
    \includegraphics[width=0.33\textwidth]{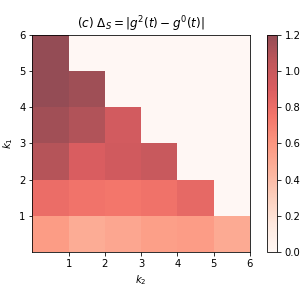} \\
    
    \caption{The results for the mixed settings in Bitcoin OTC. The average strength of an indirect attack is small and significantly smaller that the average strength of a direct attack. $\Delta_1$ shows the influence of the attack with $k_1$ direct edges, $\Delta_2$ shows the influence of the attack with $k_2$ indirect edges, $\Delta_s$ shows the influence of the attack with $k_1$ direct edges and $k_2$ indirect edges. }
    \label{fig:deltaotc}
\end{figure*}


\section{Manipulating a node directly}
\label{sec:directattack}
Let us now focus on attacks on attacking a single node.
First, we report a result on the scale of manipulability of the $\FGA$ function by a direct attack. In particular, the following theorem says that it does not take many edges to change the sign of a single node in the $\DNR$ problem when the attacker is able to rate the target directly.

\begin{restatable}{thm}{direct}
\label{thm:direct}
Let us consider an instance of the $\DNR$ problem, where, for an arbitrary $G = (V, E, \omega)$, a single node $u_T$ is attacked with $0 < g(u_T) \leq 1$ (thus $T = \{u_T\}$), and the set of attacking nodes $A \subseteq V$ is relatively trusted,  $f(v) \geq \frac{1}{2}$ for $v \in A$ and $|A| > \lceil 2 \times g(u_T) \times indeg(u_T) \rceil$. Then, it is trivial to change the sign of the goodness value of $u_T$ (i.e. to achieve the threshold $t=0$) if the attackers can attack directly (i.e. $T \subseteq I$ in the $\DNR$ problem).
\end{restatable}
The proof can be found in the appendix of the paper.

\section{Bounding the strength of indirect attacks}
\label{sec:indirectattack}
In this section, we give bounds on the strength of an indirect \textit{Sybil attacks}, i.e., the attack in which the attacker creates a new node when adding a new edge.\footnote{Our analysis also provides some additional bound for a direct Sybil attack.}
Our results hold for a family of relatively dense networks, $G=(V,E,\omega)$, in which every node has a lower bound on its indegree and outdegree, i.e., $\forall_{v \in V}{indeg(v) \geq k \And outdeg(v) \geq k}$, and the intermediary nodes, $j \in V$, are relatively weakly rated, i.e. $\sum_{v \in \pred(j)} |\omega_{vj}| \leq k$. We call such networks \textit{minimum-$k$-neighbour networks}. 

\begin{restatable}[Indirect Sybil attack]{thm}{indirect}
\label{indirect}
Assume a WSN $G=(V,E,\omega)$, where a new Sybil node $s_i$ is added that rates some intermediary node $i \neq t$. Whenever $\forall_{v \in V} indeg(v) \geq k \And outdeg(v) \geq k$, and $\forall_{j \in V} \sum_{v \in \pred(j)} |\omega_{vj}| \leq k$, then $|\Delta g(t) | \leq \frac{2}{(indeg(i)+1) \times k}$.
\end{restatable}
\begin{proof}
Let us define set $V'$ as follows. We begin with $V' = \{t\}$. Next, we iteratively add to $V'$ other nodes $v \in V$ which are indirectly connected to at least one node in $V'$ (i.e., $\exists{v' \in V'}: \exists{(l,v'), (l,v) \in E}$).
It is easy to see that the intermediary node $i$ has to belong to $V'$ in order to make the indirect attack successful.

We denote by $g(l)/f(l)$ the goodness/fairness value of the node $l$ before the Sybil attack, and by $g'(l)/f'(l)$ the goodness/fairness value after the attack. Let $\Delta g(l) = g'(l) - g(l)$ and $\Delta f(l) = f'(l) - f(l)$.

For the target node, $t \in V'$, we can calculate how its $g(t)$ changes w.r.t. the changes introduced to the goodness of all other nodes. Here, we assume that the Sybil attack is indirect, i.e., the Sybil edge is not added to $t$.
\begin{eqnarray}
g(t) = \frac{1}{indeg(t)} \sum_{u \in \pred(t)}  f(u)\times \omega(u,t)=\frac{1}{indeg(t)} \times \nonumber \\
 \sum_{u \in \pred(t)}  \hspace{-0.1cm} \omega(u,t) \times \Big[ 1-\frac{1}{outdeg(u)}\sum_{v \in \succe(u)} \frac{|\omega(u,v)-g(v)|}{2} \Big] \nonumber
\end{eqnarray}

Thus, from the triangle inequality:
\begin{eqnarray}
|\Delta g(t)| \leq \frac{1}{2 \times indeg(t)} \sum_{u \in \pred(t)} \sum_{v \in \succe(u)} \nonumber \\ \frac{1}{outdeg(u)} \times |\omega(u,t)| \times | \Delta g(v) | \leq \nonumber \\
\frac{1}{2 \times indeg(t)} \sum_{v: \exists{(t,u) \in E \And (u,v) \in E}} \sum_{u \in \succe(v)} \frac{|\Delta g(v)|}{outdeg(v)} \nonumber
\end{eqnarray}
And because $indeg(t) \geq k$, then: 
\begin{eqnarray}
\label{deltag}
|\Delta g(t) | \leq 
\frac{\sum_{v: \exists{(t,u) \in E \And (u,v) \in E}} |\Delta g(v)|}{2 \times k}.
\end{eqnarray}

Let us calculate $\Delta g(l)$ for all $l \in V'$. We can see that whenever we introduce a new node, $s_i$, that aims at node $i$ in the network, then:
\begin{eqnarray}
\label{assessment}
|\Delta g(i)| = \Big| \frac{ -1 + \sum_{u \in \pred(i)} f'(u)\times \omega(u,i) }{indeg(i)+1} - \nonumber \\ \frac{ \sum_{u \in \pred(i)} f(u)\times \omega(u,i) }{indeg(i)}  \Big| \leq \nonumber \\ \Big|
\frac{\sum_{u \in \pred(i)} \Delta f(u)\times \omega(u,i) }{indeg(i)} - \frac{1}{indeg(i+1)} - \nonumber \\  \sum_{u \in \pred(i)} \frac{f'(u) \times \omega(u,i)}{indeg(i)(indeg(i)+1)} \Big| \leq  \nonumber\\
\Big| \frac{ \sum_{u \in \pred(i)} \Delta f(u)\times \omega(u,i) }{indeg(i)} \Big| + \frac{2}{indeg(i)+1} \nonumber
\end{eqnarray}
For the other nodes, $l \in V'$, that are not targeted by $s_i$:
\begin{eqnarray}
|\Delta g(l)| = \Big| \frac{\sum_{u \in \pred(l)} \Delta f(u)\times \omega(u,l) }{indeg(l)}  \Big| \nonumber
\end{eqnarray}

For all $l \in V'$, we can write:
\begin{eqnarray}
\label{coefficients}
\Big|\sum_{u \in \pred(l)}\frac{ \Delta f(u)\times \omega(u,l) }{indeg(l)}  \Big| \leq \nonumber \\
\frac{1}{2 \times indeg(l)} \sum_{v \in \pred(l), u \in \succe(v) \setminus \{l\}} \frac{|\omega_{vl}|\times|\Delta g(v)|}{outdeg(v)} = \nonumber \\
\frac{1}{2 \times indeg(l)} \sum_{i \in V} \sum_{(v,l), (v,i) \in E} \frac{|\omega_{vi}|}{outdeg(v)}|\Delta g(i)| \nonumber
\end{eqnarray}
In the matrix form, we thus have:
\begin{eqnarray}
Q \leq M \times Q + \begin{bmatrix}
0 &
\frac{2}{indeg(i)+1} &
0&
0 &
0
\end{bmatrix}^T, \nonumber
\end{eqnarray}
where $Q$ is a vector of length $|V'|$ which on the $l$'th position has $\Delta g(l)$ for $l \in V'$. And $M$ is a matrix of size $|V'|\times|V'|$, and its coefficients are filled according to Equation~\ref{coefficients}.
Note that:
\begin{eqnarray}
\label{inftymetric}
\frac{1}{2 \times indeg(l)} \sum_{v \in \pred(l), u \in \succe(v) \setminus \{l\}} \frac{|\omega_{vl}|}{outdeg(v)} \leq \frac{1}{2}
\end{eqnarray}
This implies that $||M||_{\infty} \leq \frac{1}{2}$.
On the other hand, for a given column $j$ in the matrix $M$:
\begin{eqnarray}
\sum_{l\in V} \frac{1}{2 \times indeg(l)} \sum_{(v,l), (v,j) \in E} \frac{|\omega_{vj}|}{outdeg(v)} \leq \nonumber \\
\frac{1}{2k}\sum_{l\in V} \sum_{(v,l), (v,j) \in E} \frac{|\omega_{vj}|}{outdeg(v)} \leq 
\frac{1}{2k} \sum_{v \in \pred(j)} |\omega_{vj}| \leq \frac{1}{2} \nonumber
\end{eqnarray}
whenever $\sum_{v \in \pred(j)} |\omega_{vj}| \leq k$, which implies that $||M||_1 \leq \frac{1}{2}$.

The values of $\Delta g(l)$ achieve maximum when:
\begin{eqnarray}
Q = M \times Q + \begin{bmatrix}
0 &
\frac{2}{indeg(i)+1} &
0&
0 &
0
\end{bmatrix}^T \nonumber
\end{eqnarray}
But in this case, we can solve the equation system with:
\begin{eqnarray}
Q = \frac{1}{I-M} \times \begin{bmatrix}
0 &
\frac{2}{indeg(i)+1} &
0&
0 &
0
\end{bmatrix}^T \nonumber
\end{eqnarray}
Matrix $M$ is indeed invertible due to appropriately selected nodes $l \in V'$. What is more, since $||M||_{\infty} \leq \frac{1}{2}$, then we can write $\frac{1}{I-M} = I + M + M^2 + \ldots$ \cite{cambridge}. Finally, the above quality and $||M||_{1} \leq \frac{1}{2}$ imply that $|\sum_{l \in V} \Delta g(l)| \leq \frac{4}{indeg(i)+1}$.

Now, because Equation~\ref{deltag} holds for the target node, $t$, and $|\sum_{l \in V} \Delta g(l)| \leq \frac{4}{indeg(i)+1}$, then:
\begin{eqnarray}
|\Delta g(t) | \leq \frac{2}{(indeg(i)+1) \times k}. \nonumber
\end{eqnarray}
\end{proof}

The above theorem shows that in a minimum-$k$-neighbour network, the indirect attack is at least $k$ times weaker than the direct attack. That is, when we modify the goodness value of some node $i$ by $\Delta$, then the value of the target node $t$ is modified by at most $\frac{\Delta}{k}$.

Building upon the above reasoning, we can show that the following result also holds (the proof in the appendix of the paper).

\begin{restatable}[Direct Sybil attack]{thm}{directtwo}
\label{thm:direct2}
Assuming in a WSN $G=(V,E,\omega)$ where one adds a new Sybil node rating directly some target node $t$, then the goodness value of the target node $t$ decreases by at most $|\Delta g(t)| \leq \frac{2}{indeg(t)}$
\end{restatable}

\section{Simulations}
\label{sec:simulations}

We conduct a series of simulations on the Bitcoin OTC, Bitcoin Alpha, and RFA Net datasets studied by Kumar et al.~(\citeyear{kumar2016edge}). They consist of weighted signed networks with $|V| =\geq 3,700$, $|E| \geq 24,000$ each, where the proportion of positively weighted edges is $\geq 84\%$. A vast majority of the nodes in each network, i.e., more than $76\%$, have an indegree up to $10$.
Furthermore, most of the users in the networks are evaluated as fair by the $\FGA$ function---$f(v) \geq 0.7$ for $100\%$ of the nodes (with the mean $f(v)$ equal to $0.94$). As for goodness, only less than $4\%$ of users have a strongly positive score of more than $0.5$, and in the Bitcoin OTC network $8\%$ of users have
negative score of less than $-0.3$, whereas in Bitcoin Alpha $3,8\%$ have goodness below $-0.3$.

We focus on the attacks that lower the \emph{goodness} of the nodes, as in the $\DNR$ problem. In particular, each experiment was conducted on the set of attacking nodes $A$ of size $k = \{1,\ldots,7\}$
and the target set $T=\{t\}$ of size $1$. The target, $t \in T$, was chosen randomly from those nodes that have relatively high goodness ($g(t) \geq 0.50$) and a low indegree ($0 < indeg(t) < 10$). We study two types of the attackers:
\begin{itemize}
    \item 
    \textbf{not-established} attackers --- chosen from relatively newly created nodes with $0 < indeg < 10$ and $outdeg= 0$). This allows for studying Sybil-style attacks; and
    \item \textbf{established} attackers --- chosen from the nodes with $outdeg(v) > 5$ (and iteratively choosing nodes with fairness $f(v) > 0.7$). This allows for studying attacks by the nodes whose standing in the network has been built for some time. 
\end{itemize} 
We simulate three types of attacks:
\begin{itemize}
    \item \textbf{direct attacks} --- a set of attackers $A$ of size $k$ rates directly the target node $t \in T$. The pseudocode is presented in Algorithm~1. Each attacker rates $t$ using weight $-1$; 
    \item \textbf{indirect attacks} ---  the attackers set $A$ of size $k$ rates the neighbors of the neighbors of the target node, to minimize the goodness part of the $\FGA$ of the target node by manipulating \emph{fairness} of the targets' neighbors. The pseudocode of the attack is presented in the Alorithm~2. More precisely, the algorithm implements a greedy approach, where each new edge is used to minimize the \emph{goodness} of the target node $t$ by minimizing (or maximizing) the \emph{fairness} of one of the targets' neighbors by directly rating the successor of the target's neighbor with an edge of weight $1$ or $-1$. The algorithm performs calculations iteratively on the attackers sorted by the value of their fairness value.
    \item \textbf{mixed attack} --- $k_1$ attacking nodes perform a direct attack and $k_2$ perform an indirect one, where $k_1+k_2=k$.
\end{itemize}

\makeatletter
\newcommand{\algorithmstyle}[1]{\renewcommand{\algocf@style}{#1}}
\makeatother

The results in Figure~\ref{fig:dg} are presented with a 95\% confidence interval (marked with the opaque region around the solid/dashed lines). They show how a direct/indirect attack by established/not established nodes influences the goodness of the target node ($\Delta$). For Bitcoin OTC and Bitcoin Alpha, and RFA Net in both cases (direct and indirect attacks), there is no significant difference between the established and not established results (solid and dashed lines).

In Figure~\ref{fig:deltaotc} (see the full version in Figure~\ref{fig:deltaa}), we present results for a mixed setting. The individual cells of the heatmaps show: (a) $\Delta_1$ --- the absolute change of the goodness of the target node introduced by the $k_1$ direct edges; (b) $\Delta_2$ --- the absolute change of the goodness of the target node introduced by the $k_2$ indirect edges; and (c) $\Delta_S$ --- the total change, i.e., $\Delta_S = \Delta_1 + \Delta_2$. The results show that the average strength of a direct attack varies between $0.2$ and $1.2$ for different $k$, and the average strength of the indirect attack is lower than $0.05$, i.e., significantly smaller than the average strength of a direct attack. Theorem~\ref{thm:direct} proved in the appendix of the paper provides some  intuition why the direct attacks have such a strong impact, whereas Theorem~\ref{thm:weakinfluence} gives another intuition why the indirect attack is weaker than the direct attack.

\section{Better heuristic for indirect attacks}
\begin{algorithm}[t]
 \KwData{$A, T = \{t\}, G$}
 sort nodes in $A$ by their fairness score\\
 \While{$i <  len(sorted(A))$}{
    $a \leftarrow sorted(A)[i]$\\
   $N_1 \leftarrow \pred(t)$\\
    find a neighbor $n_2 \in \succe(n_1) \setminus \{t\}$ of a neighbor $n_1 \in N_1$ that minimizes the goodness value of $t$, when adding  an edge $(a, n_2)$ with weight $\omega(a, n_2) = 1$ or $\omega(a, n_2) = -1$\\
    $edges\_len \leftarrow$  $min(SCALE*len(indegree(n_2), MAX, len(A)-i)$ \\
    add $edges\_len$ edges to the graph $G$\\
    $i \leftarrow i + edges\_len$
 }
 \caption{Modified indirect attack}
 \label{alg:indirectmodified}
\end{algorithm}

We conduct additional tests to analyze the strength of the indirect attacks. In Algorithm~\ref{alg:indirectmodified}, instead of adding only a single edge in each iteration, as in Algorithm~\ref{alg:indirect}, we add a series of new edges. In more detail, we add $SCALE=5$ times more Sybil edges than the indegree of the target node (but at most some predefined maximum $MAX=10$). 
We take this approach to scale up the effect of manipulating the goodness value of the target nodes.

\begin{table}[t]
\centering
\begin{tabular}{llll} \hline
 & B. OTC & B. Alpha & RFA Net
\\ \hline
max $indeg$ & $10$ &  $13$ & $10$  \\ 
min $goodness$ & $0.8$ &  $0.5$ & $0.5$ \\
num of samples & $20$ & $30$ &  $27$ \\
num of edges & $20$ & $20$ &  $20$\\
\hline
\end{tabular}
\caption{The parameters used to search weak target nodes in the test sets. ``B.'' stands for ``Bitcoin''.}
\label{tab:params}
\end{table}

We attack only nodes with bounded $indegree$, and the goodness value bigger than some threshold. We believe these nodes are more easily manipulable than an average node. See Table~\ref{tab:params} for the details.

\begin{table}[H]
\centering
\begin{tabular}{llll} \hline
 & B. OTC & B. Alpha & RFA Net \\ \hline
average change & $0.081$ &  $0.085$ & $0.030$  \\ 
standard deviation & $0.089$ &  $0.085$ & $0.028$ \\
min change & $0.010$ &  $0.008$ & $0.009$ \\
max change & $0.300$ &  $0.298$ & $0.131$ \\
median & $0.053$ &  $0.042$ & $0.021$ \\
0.75-quantile & $0.111$ &  $0.121$ & $0.041$ \\
\hline
\end{tabular}
\caption{The results of Algorithm~\ref{alg:indirectmodified} on different datasets.} 
\label{tab:thres}
\end{table}
The analysis of the data in Table~\ref{tab:thres} shows that the attack using the Algorithm~\ref{alg:indirectmodified} may (but rarely does) achieve relatively strong results in some cases. To be more precise, the maximum change of the goodness value of the target node introduced by the indirect attack in the Bitcoin OTC and Bitcoin Alpha detasets reached the barrier of $0.3$. This shows that in general networks (unlike the minimum-$k$-neighbour ones described in Theorem~\ref{indirect}) do not have strong resistance property against indirect attacks.
In most cases however the attack gives rather weak results ($75\%$-quantile on all datasets is at most $0.12$ with low median of at most $0.05$). The minimum strength of the attack in all datasets achieves $0.01$.

\section{Conclusions}
In this paper, we axiomatized the $\FGA$ measure with respect to, among others, the properties of homogeneously and unanimously rated nodes and with respect to the properties of the rating nodes that achieve constant rating error. Furthermore, we presented the hardness results on the manipulability problems. We also derived analytical results concerning the strength of the direct attacks and weakness of the indirect attacks in the networks in which each node has minimum $k$ neighbours (in and out). Finally, we visualised experimentally the strength of direct attacks and analysed two different greedy algorithms for indirect attacks. This showed that $\FGA$ might be manipulated indirectly in non-minimum-$k$-neighbour networks. Overall, a higher-level insight from our analysis is that $\FGA$ is generally more difficult to manipulate compared to other social network analysis tools (e.g., centrality measures). In particular, while worst-case hardness results are common in the literature, various other tools turned out to be easily manipulable in practice by well-crafted heuristics \cite[e.g.]{bergamini2018improving,Waniek2018,waniek2019hide}. The $\FGA$ measure turns out to be more resilient, which provides a good argument for using it in practice.

In future it would be plausible to compare other candidate measures existing in the literature in the terms of manipulability and try to derive a more general approach for analyzing the manipulability of weighted ranking functions. We also encourage studies on the axiomatization of the ranking functions, which would result in a better understanding of their properties.

\nobibliography*
\bibliography{aaai23}



\clearpage
\appendix

{
\large
\noindent \textbf{An appendix to ``Predicting Weights in Signed Weighted Networks is Difficult to Manipulate'':
}
}
\bigskip
\section{Axiomatization}

\subsection{Goodness axiomatization}

Below we formally define the axioms and present the uniqueness proofs.

\noindent We begin with the characterization of goodness part of the $\FGA$ function.  
Let $v \in V$ have all the predecessors $u_i \in \pred(v)$ homogenous and unanimous w.r.t. $v$, i.e. they all have the same fairness $f_0$ and they rate $v$ with the same rating $\omega_0$. Now, let us assume that $f_0$ of all the predecessors gets increased by the same amount, $\Delta$. We require that the goodness of the rated node $v$ should rise proportionally to $\Delta$ (see Figure~\ref{fig:axiom1}).
%
%
%
To formalize this axiom, let us denote the goodness of $v$ in such a setting by $g^{\phi_{\omega}, \phi_{f}}(v)$, where $\phi_{\omega}$  indicates the value of weight of the edges $(u_i, v)$, and $\phi_{f}$ indicates the value of the fairness of all $u_i \in \pred(v)$.

\begin{axiom}[SMOOTH GOODNESS]\label{axiom:increase:fairness:def}
Let $v \in V$, such that  $\forall u_i\hspace{-0.075cm} \in \pred(v) \ f(u_i) = f_0 \land\ \omega(u_i,v) = \omega_0$. 
Then $\forall \Delta \in \mathbb{R}$: 
\[g^{\omega_0, f_0+\Delta}(v) = g^{\omega_0, f_0}(v) + g^{\omega_0, \Delta}(v).\]
\end{axiom}

Next, let us consider an analogous situation, but now the weight $\omega_0$ of the edges from the predecessors $\pred(v)$ to $v$ increases by $\Delta$ while their fairness $f_0$ remains the same (see Figure~\ref{fig:axiom2}). This leads to the following axiom:

\begin{axiom}[INCREASE WEIGHT]\label{axiom:increase:weight:def}
Let $v \in V$, such that  $\forall \hspace{-0.025cm} u_i \hspace{-0.075cm} \in \hspace{-0.075cm} \pred(v)  f(u_i) = f_0 \land \omega(u_i,v) = \omega_0$. Then, $\forall\hspace{-0.025cm} \Delta \hspace{-0.025cm} \in \hspace{-0.025cm} \mathbb{R}$: 
\[g^{\omega_0+\Delta, f_0}(v) = g^{\omega_0, f_0}(v) + g^{\Delta, f_0}(v).\]
\end{axiom}

Next, we require that nodes with higher fairness have a higher impact on the goodness of the rated nodes. Similarly, higher weights should result in a better rating of the target node (see Figure~\ref{fig:axiom3}).

\begin{axiom}[MONOTONICITY FOR GOODNESS] \label{axiom:monotonous:goodness:def}
Let $u_1$ and $u_2$ be two nodes rated by unanimous and homogeneous sets of predecessors $S_1$, $S_2$. $S_i$ consisting of nodes with identical fairness $f_{i}$ who rate $u_i$ with identical $\omega_{i}$.
Then, if $f_1 = f_2$ and $\omega_{1} > \omega_{2}$, then $g(u_1) \geq g(u_2)$. Also, if $\omega_1 = \omega_2$ and $f_{1} > f_{2}$, then $g(u_1) \geq g(u_2)$ as well.
\end{axiom}

MONOTONICITY FOR GOODNESS is a weaker version of the Goodness Axiom proposed by Kumar et al.~(\citeyear{kumar2016edge}). While the Goodness Axiom concerns any predecessors, MONOTONICITY FOR GOODNESS focuses on unanimous and homogeneous sets of them.

Next, any node, $v \in V$, that has the best possible rating given by each of its predecessors and all its predecessors have the highest possible fairness, then $v$ should have maximal possible goodness (see Figure~\ref{fig:axiom3}).
\begin{axiom}[MAXIMAL TRUST]\label{axiom:maximal:trust}
For any $v \in V$ such that  $\forall u_i \in \pred(v) \ \ f(u_i) = 1 \ \land\ \omega(u_i,v) = 1$, it holds that, $\forall \Delta \in \mathbb{R}, \ g(v)=1.$
\end{axiom}

The following axiom states that, when $v \in V$ is rated by $k$ groups, where the nodes in each group are homogeneous and unanimous w.r.t. $v$, then the goodness of $v$ should be equal to the \emph{weighted average} of the ratings achieved when these groups separately rate $v$.

\begin{axiom}[GROUPS FOR GOODNESS]\label{axiom:groups}
Given $v \in V$, let $\{S_1,\ldots,S_k\}$ be a partition of $\pred(v)$  such that $\forall i \in [k]$ there exists $f_i,\omega_i: \forall_{u_j \in S_i} f(u_j) = f_i \land \omega(u_j, v) = \omega_i$. Then, it holds that:
\[g(v) = \frac{\sum_{i \in [k]}(|S_i|\times g_{i}(v))}{\sum_{i \in [k]}|S_i|},\] 
\noindent where $g_i(v)$ denotes the rating of the node $v$ rated only by the homogeneous and unanimous predecessors from group~$i$. 
\end{axiom}

Finally, we have the following baseline:

\begin{axiom}[BASELINE FOR GOODNESS]\label{axiom:baseline}
Any $v$ with $indeg(v) = 0$ has $g(v) = 1$.
\end{axiom}

\begin{figure}[t]
    \centering
    \includegraphics[width=0.5\textwidth]{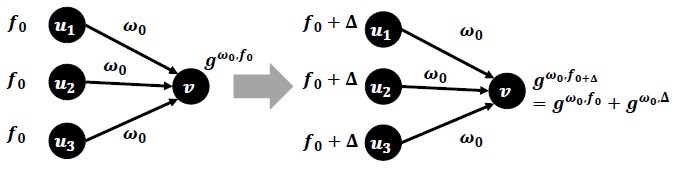}
    \caption{Axiom~\ref{axiom:increase:fairness:def} says that the increase in the homogenous \emph{fairness} of the unanimous predecessors results in the proportional increase of the \emph{goodness} of the rated node.}
    \label{fig:axiom1}
\end{figure}

\begin{figure}[t]
    \centering
    \includegraphics[width=0.5\textwidth]{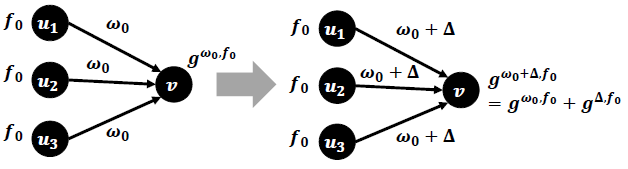}
    \caption{Axiom~\ref{axiom:increase:weight:def} says that the increase in the homogenous \emph{weight} of the unanimous predecessors results in the proportional increase of the \emph{goodness} of the rated node.}
    \label{fig:axiom2}
\end{figure}

We will now show that the above axioms uniquely define the goodness part of the $\FGA$ function.
\begin{restatable}{thm}{goodness}
\label{thm:goodness}
For any fixed fairness function $f(u)$, the SMOOTH GOODNESS, INCREASE WEIGHT, MONOTONICITY FOR GOODNESS, MAXIMAL TRUST, GROUPS FOR GOODNESS, and BASELINE FOR GOODNESS axioms uniquely define goodness function~\eqref{def:goodness}.
\end{restatable}
\begin{proof}
It is easy that the goodness function~\eqref{def:goodness} meets the conditions of the above axioms.
Now, let us define some $g_i^{w_0,f_0}(v)$ for a node $v$ rated by homogeneous and unanimous nodes w.r.t. $v$, i.e. all $u \in \pred(v)$ have the same \emph{fairness} $f(u) = f_0$ and they rate $v$ with the same rating $\omega(u,v) = \omega_0$,
From SMOOTH GOODNESS and MONOTONICITY FOR GOODNESS and the Cauchy's equation \citep{functionaleqn}, we know that $g_i(v)$ is linearly dependant on $f(u)$ when $\omega(u,v)$ is fixed to some $\omega_0$, i.e. $g_i^{\omega_0,f(u)}(v) = a_{\omega_0}\times f(u)$ for some constant $a_{\omega_0} \in \mathbb{R}$. Again, from INCREASE WEIGHT, MONOTONICITY FOR GOODNESS, and the Cauchy's equation we know that $g_i(v)$ is linearly dependant on $\omega(u,v)$ when $f(u)$ is fixed to some $f_0$, i.e. $g_i^{\omega(u,v), f_0}(v) = \omega(u,v)\times b_{f_0}$.
The two equations above imply that for a set of homogeneous and unanimous predecessors, $g_i^{\omega(u,v), f(u)}(v) = g_i^{\omega, f}(v) = a_{\omega}\times f = b_{f}\times \omega$. Since the function $g_i^{\omega(u,v), f(u)}(v)$ is defined for all $\omega, f \in \mathbb{R}$, this equality implies that $a_{\omega} = \frac{b_{f}}{f}\times \omega$. Furthermore, since $a_{\omega}$ is not dependant on $f$ by definition, then $a_{\omega} = c \times \omega$ for some $c \in \mathbb{R}$. We conclude that $g_i^{\omega, f}(v) = c\times f\times \omega$.
From MAXIMAL TRUST, we get that $g^i(v) =f(u)\times \omega(u,v)$. Now, when a node does not have unified predecessors, we can divide its predecessors to groups with fixed $(f_i, \omega_i)$. From GROUPS FOR GOODNESS, we get: \[ g(v) = \frac{\sum_{i \in [k]}(|S_i|\times g_{i}(v))}{\sum_{i \in [k]}|S_i|} = \frac{\sum_{i \in [k]}(|S_i|\times f_i\times \omega_i)}{\sum_{i \in [k]}|S_i|} = \] \[= \frac{1}{in(v)} \sum_{u \in \pred(v)} f(u) \times \omega(u,v). \] From BASELINE FOR GOODNESS, $g(v) = 1$ for $v$ with $\indeg(v) = 0$.
\end{proof}

\begin{figure}[t]
    \centering
    \includegraphics[width=0.45\textwidth]{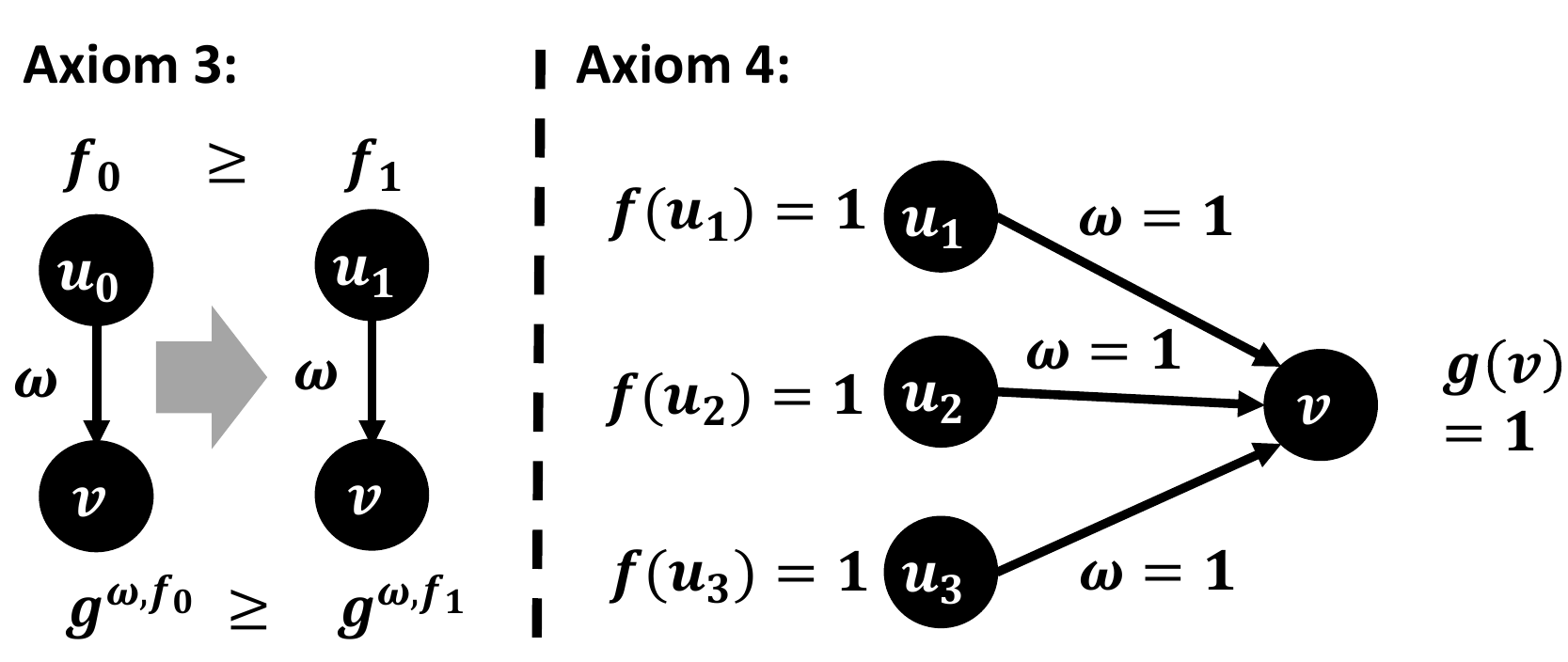}
    \caption{Axiom~\ref{axiom:monotonous:goodness:def} says that fairer nodes should have a bigger impact on the goodness of the rated nodes. In the figure, we assume that $\omega > 0$ and node $u_0$ has bigger fairness than node $u_1$. Axiom~\ref{axiom:maximal:trust} says that a fully trusted node should have the goodness value equal to 1.}
    \label{fig:axiom3}
\end{figure}

\subsection{Fairness axiomatization}
In this section, we present the axiomatization of the fairness part of the $\FGA$ function. The fairness axiomatization is defined with respect to the rating error of nodes. We define the error of node $v$ rating the node $u$ as $d=|\omega(v,u)-g(u)|$.

Our first axiom stipulates that the fairness of a node that makes an average error when rating other nodes is equal to the average of the fairness values of nodes in extreme cases.

\begin{axiom}[SMOOTH FAIRNESS]\label{axiom:smooth:fairness}
Assume a node $v$ rates a set of its successors $S$ with equal error $d  = |g(u)-\omega(v,u)|$ for $u \in S$ in one setting, and with an error $D$ in another setting, then $f^{\frac{d+D}{2}}(v) = \frac{f^{d}(v) + f^{D}(v)}{2}$.
\end{axiom}

The following axiom states that fairness of the nodes that rate more accurately should rise (see Figure~\ref{fig:axiom8}).

\begin{figure}[t]
    \centering
    \includegraphics[width=0.32\textwidth]{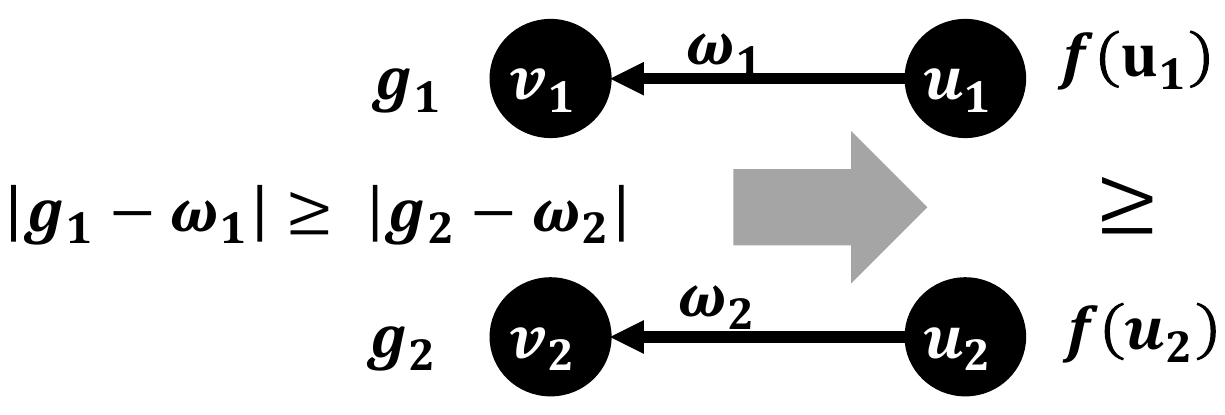}
    \caption{Axiom~\ref{axiom:monotonous:fairness:def} says that the fairness of a node should rise when the node gives more precise ratings.}
    \label{fig:axiom8}
\end{figure}

\begin{axiom}[MONOTONICITY FOR FAIRNESS ] \label{axiom:monotonous:fairness:def}
Let $u_1$ and $u_2$ be two nodes rating their sets of successors $S_1$, $S_2$. $S_i$ consists of nodes $v_i$ rated by $u_i$ with identical error $d_i  = |g(u_i)-\omega(u_i,v_i)|$.
If $d_1 > d_2$, then $f(u_1) \leq f(u_2)$.
\end{axiom}
This is a weaker version of the Fairness Axiom in \cite{kumar2016edge}. In our case, it is defined only for a set of successors $S_i$ rated with \emph{equal} rate by the node $u_i$.

Next, we stipulate that when a node makes maximal errors when rating all of its neighbors, then its fairness should be $0$, and when it always agrees with the actual goodness value of its rated nodes, then its fairness is $1$ (see Figure~\ref{fig:axiom9}).

\begin{axiom}[OBVIOUS FAIRNESS METRIC]\label{axiom:obvious:fairness}
Assume node $v$ rates all its successor nodes $S$ with distance $d = |g(u)-\omega(v,u)| = 0$, for $u \in S$, then $f(v) = 1$. Assume a node $v$ rates all its successor nodes $S$ with distance $d = |g(u)-\omega(v,u)| = 2$, for $u \in S$, then $f(v) = 0$. 
\end{axiom}

Also, when $v \in V$ rates its neighbors that can be divided to $k$ such groups that each node in a group is rated by $v$ with the same distance as other nodes in this group, then the fairness of $v$ should be equal to the \emph{weighted average} of its fairness in a setting where $v$ rates these groups separately.

\begin{axiom}[GROUPS FOR FAIRNESS]\label{axiom:groups:fairness}
Given $v \in V$, let $\{S_1,\ldots,S_k\}$ be a partition of $\succe(v)$  such that $\forall i \in [k]$ there exists $d_i: \forall_{u_j \in S_i} |g(u_j)-\omega(v,u_j)| = d_i$. Then:
\[f(v) = \frac{\sum_{i \in [k]}(|S_i| \times f_{i}(v))}{\sum_{i \in [k]}|S_i|},\] where $f^i(v)$ is the fairness of $v$ rating group $i$.
\end{axiom}
Finally, a baseline for node $v$ with $outdeg(v) = 0$ is:
\begin{axiom}[BASELINE FOR FAIRNESS]\label{axiom:baseline:fairness}
A node $v$ with $outdeg(v) = 0$ has $f(v) = 1$.
\end{axiom}
\begin{restatable}{thm}{fairness}
\label{thm:fairness}
For fixed goodness function $g(n)$, the SMOOTH FAIRNESS, MONOTONICITY FOR FAIRNESS , OBVIOUS FAIRNESS METRIC, GROUPS FOR FAIRNESS, and BASELINE FOR FAIRNESS axioms uniquely define fairness function \eqref{def:fairness}.
\end{restatable}
\begin{proof}
It is easy that the fairness function~\eqref{def:fairness} meets the conditions of the above axioms.
Now, let us define $f_i(v)$ for a node $v$ and a group of nodes with some fixed error $d_i =|g(u)-\omega(v,u)|$,
From SMOOTH FAIRNESS, MONOTONICITY FOR FAIRNESS  and the Jensen's equation \citep{functionaleqn}, we know that $f_i(v)$ is linearly dependant on $d_i = |g(u)-\omega(v,u)|$, i.e. $f_i(v) = b+a \times d_i$ for some $a,b \in \mathbb{R}$.
From OBVIOUS FAIRNESS METRIC we get that $f_i(v) =1-d_i/2=1-|g(u)-\omega(v,u)|/2 $. Now when a node does not have unified successors, we can divide its successors to groups with fixed $|g(u_j)-\omega(v,u_j)| = d_i$. From GROUPS FOR FAIRNESS: \[ f(v)= \frac{\sum_{i \in [k]}(|S_i| \times f_{i}(v))}{\sum_{i \in [k]}|S_i|} = \frac{\sum_{i \in [k]}(|S_i| \times (1-d_i/2))}{\sum_{i \in [k]}|S_i|} = \] \[ = 1- \frac{1}{out(v)} \sum_{u \in Succ(v)} |g(u)-\omega(v,u)|/2. \]
Finally, from BASELINE FOR FAIRNESS we get that $f(v) = 1$ for nodes with $outdeg(v) = 0$.
\end{proof}

\begin{figure}[t]
    \centering
    \includegraphics[width=0.4\textwidth]{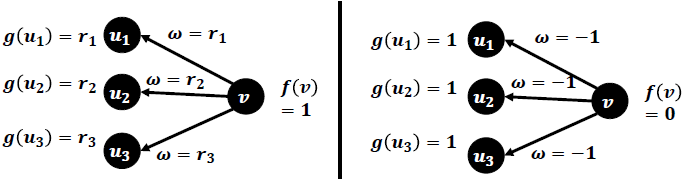}
    \caption{Axiom~\ref{axiom:obvious:fairness} says that a node that rates with perfect accuracy (its rating error $d=0$ for all the nodes that it rates) should have maximal fairness equal to $1$, and a node that makes the biggest errors (its rating error $d=2$ for all the nodes that it rates) should have minimal fairness equal to $0$.}
    \label{fig:axiom9}
\end{figure}
\subsection{FGA axiomatization}
The above results imply the final axiomatization result:
\axiomatization*
\begin{proof}
The proof follows from the proofs of Theorems~\ref{thm:goodness}~and~\ref{thm:fairness}.
\end{proof}

\section{Ommitted complexity proofs}
\vertexcover*
\begin{proof}[Proof of Theorem~\ref{thm:vertexcover}]
\begin{figure}[t]
    \centering
    \includegraphics[width=0.4\textwidth]{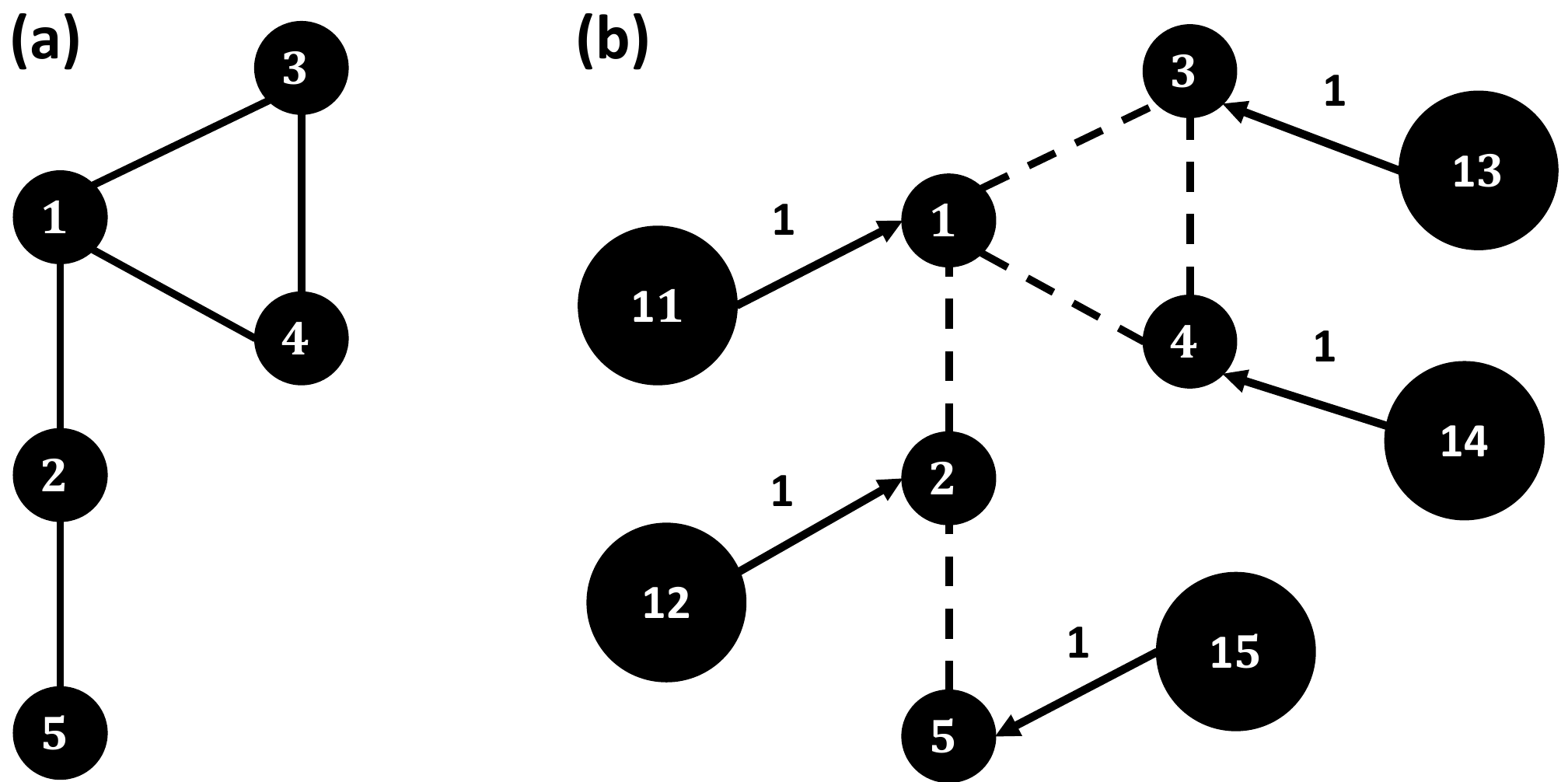}
    \caption{(a) The original VC problem ($k=2$) with vertices $V = \{1,2,3,4,5\}$ and edges as in the picture. This set can be covered with 2 nodes - $1$ and $5$. (b) The corresponding $\DMT$ instance with number of moves $k=2$, threshold $t=-1$, the set of attacked edges $H$, as in the original problem, and a new set of attacking nodes $\{11,12,13,14,15\}$ marked with big circles for which attacking edges were created $\{(11, 1), (12, 2), (13,3), (14,4), (15, 5)\}$ - each with the weight of $1$. To solve this problem we need to modify the values of the edges $\{(11, 1), (15, 5)\}$ to $-1$.}
    \label{figure:1}
\end{figure}

We reduce from the VERTEX COVER (VC) problem. In the VC problem we are given a parameter $k$ and a graph $G = (V, E)$ and we need to decide whether there exists a set of vertices, $U \subseteq V$, $|U| \leq k$, that ``cover'' the set of the edges of this graph, i.e., every edge from $E$ is adjacent to at least one node in $U$.

Given the VC problem $(G,k)$, where $G=(V,E)$, we create an instance of our problem $\DMT = (G'=(V', E', \omega), A, \TP, I,  t, k')$, by adding, for every node $v \in V$, an \emph{attacking} node $a_v$ with an edge $(a_v,v)$ with $weight = 1$. We set the target threshold as follows: $t = -1$. We observe that $A = \{a_v: v \in V\}$, $V' = V\cup A$, $E' = \{(a_v,v): v \in V\}.$ Finally, the set of  intermediary vertices is $I = V$, and the set of attacked edges $\TP$ is $E$, i.e., the set of the edges from $G$. We set $k'=k$. 
See Figure~\ref{figure:1} for an example.

\begin{figure}[t]
    \centering
    \includegraphics[width=0.30\textwidth]{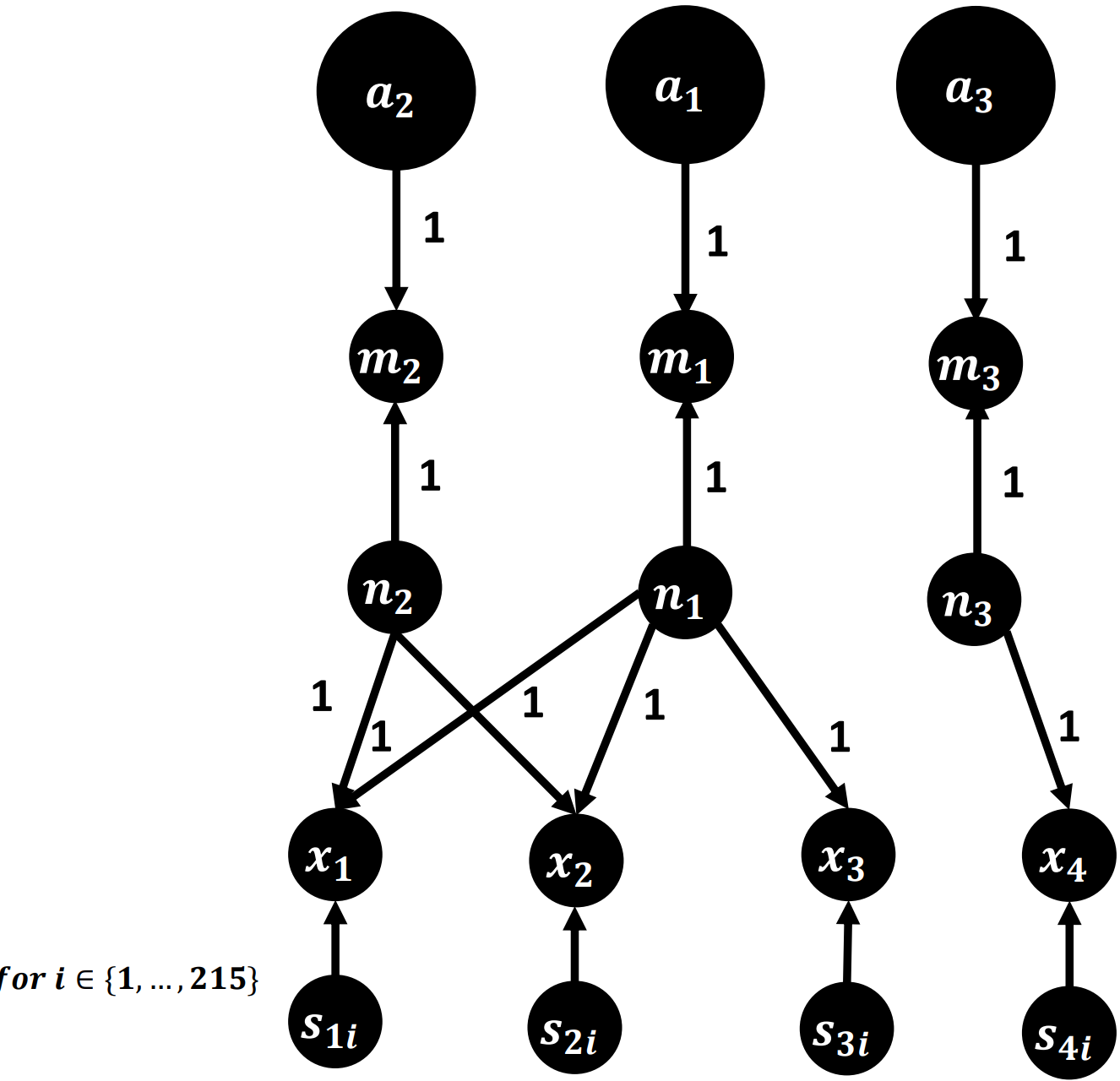}
    \caption{The corresponding $\DNR$ instance with the set of targets $\{1,2,3,4\}$, the set of attackers $\{a_1, a_2, a_3\}$. The intermediary nodes created for every set from $\mathcal{S}$, i.e. $\{n_1, n_2, n_3, m_1, m_2, m_3\}$, $215$ stabilising nodes $s_{ji}$ for each node $x_j$, $k=2$, threshold $t=1-\epsilon$.}
\label{figure:dnrred}
\end{figure}

We now need to show that the reduction is correct. Firstly, given a graph $G = (V, E)$ and its vertex cover of size $k$, i.e., $U \subseteq V$ with $U = \{x_1, ..., x_k\}$, we show that its corresponding problem, $\DMT$, as outlined above can be solved.
To this end, we modify all of the $k$ edges $(a_{x_i},x_i)$, where $x_i \in U$, by setting  each of them to $-1$. Now:
\begin{itemize}
    \item due to the fact that computing $(f(u), g(u)) = \FGA(G, u)$ of a node $u$ adjacent only to a single directed edge results in $g(u)$ being equal to the weight of this single edge, then for all $x_i \in U$ we have $g(x_i) = -1$;
    \item from the definition of the $\FGA$ function, the fairness of nodes with out-degree of $0$ is equal to $1$. Hence, for all $u \in V$ we have that $f(u) = 1$.
\end{itemize}    
From these we conclude that all of the connections in the target set $\TP$ are decreased to the threshold $t=-1$, i.e. either $f(u)\times g(v) = 1\times(-1) = -1$ or $f(v)\times g(u) =-1$ for every pair ${u,v} \in \TP$.

For the other direction, assume that we have a solution to our problem $\DMT$ that was created as outlined above. 
Recall that the set of attacking nodes is the set of newly created nodes for the $\DMT$ instance, i.e., $A = \{a_v: v \in V\}$, and each of them is connected with a single edge directed towards its corresponding node from $V$, i.e., $(a_v, v)$ for $v \in V$ and weight of these edges equals 1. 
We observe that, from the definition of $\FGA$ and the construction of the $\DMT$ instance, it follows that to modify the values of the predicted connections between pairs $\{u,v\} \in \TP$ (i.e. either $f(u)\times g(v)$ or $f(v) \times g(u)$) one needs to modify the goodness of the nodes in $\dom(\TP)$. This is because there are no outgoing edges from the nodes in $\TP$; thus, fairness of the nodes in $\dom(\TP)$ is constant and equal to 1.

The goodness of the nodes in $\dom(\TP)$ can be modified by changing $(a_v,v)$, where $a_v \in A$, or by adding some new edges between the attackers and the nodes from $\dom(\TP)$. However, to attack a single connection $\{u,v\} \in \TP$, we have to obtain either $g(u) = -1$, or $g(v) = -1$. To this end, since reaching $g(v) = -1$ is only possible when all of the  edges pointed at $v$ have value $-1$, 
it is always necessary to modify the value of the existing edge $(a_v,v)$ as well.
Specifically, whenever we can reach one of the nodes in $\dom(\TP)$ with a modified edge, we are, in fact, ``marking'' all of the edges pointing at this node. Each of these edges corresponds to a pair in $\TP$. If all the pairs are marked, then both the $\DMT$ and $\VC$ problems are solved.
\end{proof}
The proof for the $\IMT$-hardness is analogous with the  opposite signs of the weights of the created/modified edges.
\setcover*
\begin{proof}
[Proof of Theorem~\ref{thm:dnr}] We reduce from the SET COVER (SC) problem. In the SC problem, we are given a set of sets  $\mathcal{S}$, a target set $T$, and a parameter $k$. We need to decide whether it is possible to cover the target set $T$ with at most $k$ sets from $\mathcal{S}$, i.e. whether there exists a subset $S \subseteq \mathcal{S}$ of size at most $k$, such that for all $t \in T$, there exists $S_i \in S$, such that $t \in S_i$.
Given an SC problem $(\mathcal{S},T, k)$, we create an instance of the $\DNR$ problem as follows:
\begin{itemize}
\item the set of target nodes in the $\DNR$ problem is the set $T$ from the SC problem; 
\item For every set $S_i$ from $\mathcal{S}$ we create two \emph{intermediary} nodes $n_i$ and $m_i$ one link $(n_i, m_i)$ and $|S_i|$ links $(n_i, t)$ for $t \in S_i$, each of them with weight 1. We denote the set of all intermediary nodes $n_i$ which point at the nodes $m_i$ as $\mathcal{N}_{int}$;
\item for each $S_i \in \mathcal{S}$ we create an attacking node $a_i$ and a link $(a_i, m_i)$ with weight 1; and
\item we set the intermediary set $I$ in the $\DNR$ problem to the set of $m_i$ nodes,
\item given $d_{max} = max\ \{outdeg(n_i): n_i \in \mathcal{N}_{int}\}$, we add $l = 8d_{max}^3-d_{max}+1$ stabilising nodes $s_i$ to each target node $x_j \in T$, they are required in the reduction to ensure that the change in goodness of the target nodes will not affect the goodness of some other nodes too much,
\item  we set the target threshold in the $\DNR$ problem to be $1-\epsilon$, where $\epsilon = \frac{1}{4*d_{max}(d_{max}+l)}$,
\item we set the budget to $k$.
\end{itemize}

In Figure~\ref{figure:dnrred}, we present a sample $\DNR$ construction for the SC problem in which $T = \{1,2,3,4\}$ has to be covered  with at most $k=2$ sets from $\mathcal{S} = \{ \{1,2\}, \{1,2,3\}, \{4\} \}$.

We now need to show that the reduction is correct. Firstly, let us consider an SC problem $(\mathcal{S},T, k)$ and its set cover of size $k$, consisting of sets $S_i \in \mathcal{S}$ with indexes $i \in \{x_1, \ldots, x_k\} = U$. We will show that our corresponding $\DNR$ problem created as in the instructions above can be solved. To this end, we modify the value of 
each link $(a_{x_i}, m_{x_i})$ for $i \in U$ to $-1$.

In this case, the goodness of the intermediary node decreases to a value bounded by the factor introduced by $\omega((a_{x_i}, m_{x_i})) \times f(a_{x_i}) = -1 \times f(a_{x_i}) \leq 0$ and the factor introduced by $\omega((n_{x_i}, m_{x_i})) \times f(n_{x_i}) \leq 1$, resulting in a value $g(m_{x_i}) \leq \frac{1}{2}$. This implies the decrease in the fairness value of the intermediary node $n_{x_i}$, resulting in $f(n_{x_i}) \leq 1 - \frac{1-\frac{1}{2}}{2*outdeg(n_{x_i})} \leq 1 - \frac{\frac{1}{2}}{2*d_{max}}$. Finally this decreases the goodness values of all nodes $v_j$ rated by $n_{x_i}$ to a value less or equal to $g(v_j) \leq \frac{indeg(v_j)-1}{indeg(v_j)}+\frac{1}{indeg(v_j)}(1 - \frac{\frac{1}{2}}{2*d_{max}}) \leq 1 - \frac{1}{4d_{max}*indeg(v_j)} \leq 1 - \frac{1}{4*d_{max}*(l+d_{max})} \leq 1 - \epsilon$.
Since $T$ is covered by the sets indexed by indices in $U$, then  modifying the  value of the links $(a_{x_i}, m_{x_i})$ decreases the rating of all of the target nodes in the $\DNR$ problem below or to the threshold.

For the other direction, assume that we have a solution to our corresponding problem $\DNR$. In fact, since the only allowed actions are edge additions and weight updates to the nodes from $I = \{m_i\}_i$, the only way of modifying the goodness of the target nodes is by modifying the \emph{fairness} of the intermediary nodes $n_i$. Either modifying an edge $(a_i, m_i)$ or adding an edge $(a_j, m_i)$ marks a set $S_i \in \mathcal{S}$ and sets the value of the goodness of the nodes $k \in S_i$ below the threshold $1 - \epsilon$. 
One needs to see that it is necessary to rank the node $m_{i}$ to mark the node $v_j \in S_i$, otherwise its goodness value will stay above threshold (i.e. $goodness(v_j) > 1 - \epsilon$. We achieve this result by introducing $l$ stabilising nodes for every $v_j \in T$. From the properties of the given construction one may conclude that marking nodes in the $\DNR$ problems implies marking sets in the set cover problem.

We will use an intermediary Theorem~\ref{thm:weakinfluence}. It shows that for a node $x$ when fairness of its $k$ rating nodes is decreased by $\Delta$, and there are $l$ stabilising nodes rating it with $1$, then the goodness value of the node $x$ does not change too much - i.e. $g(x) \geq 1 - 2 \frac{k}{l+k}\times \Delta$.

Using this result we may see that even in an edge case the nodes in the target set do not have their \emph{goodness} value changed below the threshold if they are not marked properly as mentioned before. In the edge case a node $v_i$ may be indirectly influenced by a set of $k_1$ nodes (denoted $\mathcal{K}$), which have their fairness value indirectly changed because they rate at most $k_2$ nodes (denoted $\mathcal{L}$) which are marked by at most $k_3$ intermediary nodes which change their fairness value by at most $\Delta = 1$. Note that $k_1,k_2,k_3 \leq d_{max}$.

In this case the goodness value of the nodes in $\mathcal{L}$ can be bounded by the above theorem $g(v_j) \geq 1 - 2 \frac{d_{max}}{d_{max}+l}$. This implies that the fairness of the nodes in $\mathcal{K}$ falls to a value not less that $f(n_m) \geq 1 - \frac{d_{max}}{d_{max}+l}$. This fairness modification will further influence the target nodes, but since in any scenario also for the target nodes we have $g(v_i) \geq 1 - 2 \frac{d_{max}}{d_{max}+l}$, then the fairness value of the intermediary nodes will not fall below $1 - \frac{d_{max}}{d_{max}+l}$. Finally we can conclude that  the nodes in the target set are influenced by at most $g(v_i) \geq 1 - 2 \frac{d_{max}}{d_{max}+l}\frac{d_{max}}{d_{max}+l}$. We can see that when $l$ is big enough, this value never reaches the threshold $\epsilon$, i.e. $1 - 2 \frac{d_{max}}{d_{max}+l}\frac{d_{max}}{d_{max}+l} > 1 - \frac{1}{4*d_{max}(d_{max}+l)}$ when $l > 8d_{max}^3-d_{max}$.
\end{proof}

The proof for the $\INR$-hardness is analogous with the  opposite signs of the weights of the created/modified edges.

\begin{thm}
\label{thm:weakinfluence}
We have a node $x$ that is rated by $k$ influencing nodes $n_i$ (for $i \in [k]$). What is more this node is rated by $l$ other stabilising nodes $s_j$ (for $j \in [l]$). All rates are of value $1$. Suppose the fairness value of the influencing nodes decreases by at most $\Delta$ (after all modifications in the network), then the goodness value of the node $x$ decreases by at most $2\frac{k}{l+k}\Delta$.
\end{thm}
\begin{proof}
 By MONOTONICITY FOR GOODNESS we know that the goodness value of the node $x$ will decrease maximally when we decrease the fairness value of all influencing nodes $n_i$ by exactly $\Delta$.  We can estimate how the fairness value of the stabilising nodes ($f^{(t)}(s_i)$) and the goodness value of the rated node ($g^{(t)}(x)$) will change in the next iterations of the $\FGA$ function computation.

By the $\FGA$ definition, the stabilising nodes $s_i$ which rate only one node $x$ have $f^{(0)}(s_i) = 1$ and $f^{(t)}(s_i) = 1-\frac{|1-g^{(t-1)}|}{2}$ for $t \geq 1$. What is more since all ratings are of value $1$, we know that $g(x) \geq 0$, then $f^{(t)}(s_i) = 1-\frac{1-g^{(t-1)}}{2}$ for $t \geq 1$. The goodness value of the node $x$ rated by $k$ nodes $n_i$ with decreased fairness and $l$ stabilising nodes $s_i$, can be bounded as follows - $g^{(0)}(x) = 1$ and $g^{(2t)}(x) \geq \frac{k}{k+l}(1-\Delta)+\frac{l}{l+k} \times (1-\frac{1-g^{(2t-2)}}{2})$ for $t \geq 1$.
We prove by induction that for $2t \geq 2$ we have $g^{(2t)}(x) \geq 1 - \frac{k\Delta}{k+l}\sum_{2i=0}^{2t-2}[\frac{l}{2(l+k)}]^{2i}$. The statement trivially holds for $2i = 0$. Let's assume it holds for $2t$, then for $2t+2$ we have $g^{(2t+2)}(x) \geq \frac{k}{k+l}(1-\Delta)+\frac{l}{l+k} \times (1-\frac{1-g^{(2t)}}{2}) \geq \frac{k}{k+l}(1-\Delta)+\frac{l}{l+k} \times (1-\frac{1-[1 - \frac{k\Delta}{k+l}\sum_{2i=0}^{2t-2}[\frac{l}{2(l+k)}]^{2i}]}{2}) = 1 - \frac{k\Delta}{k+l}\sum_{2i=0}^{2t}[\frac{l}{2(l+k)}]^{2i}$ what proves the induction. 
We may also further bound this sum $g^{(2t)}(x) \geq 1 - \frac{k\Delta}{k+l}\sum_{2i=0}^{2t-2}[\frac{l}{2(l+k)}]^{2i} \geq 1 - \frac{k\Delta}{k+l}\sum_{2i=0}^{2t-2}[\frac{1}{2}]^{2i} \geq 1 - \frac{2k\Delta}{k+l}(1-(\frac{1}{2})^{t-1}) \geq 1 - 2\Delta\frac{k}{k+l} $
It is also easy to see that $g^{(2t)}(x) = g^{(2t-1)}(x)$, thus we can conclude that $g(x) \geq 1 - 2\Delta\frac{k}{k+l} $.
\end{proof}
\wparam*
\begin{proof}
[Proof of Theorem~\ref{thm:w2}]
The Set Cover problem parameterized by the number of sets $k$ is a $W[2]$-hard problem in the W-hierarchy. Since the reduction in the proof of Theorem~\ref{thm:dnr} runs  polynomial time, the budget of the $\DNR(\INR)$ problem is $k$, this reduction is also a parameterized reduction~\citep{parameterized}.
\end{proof}
\wparamdmt*
\begin{proof}
[Proof of Theorem~\ref{thm:w2dmt}]
One needs to see that a slight modification of the reduction in the proof of Theorem~\ref{thm:dnr} allows to create a parameterized reduction from the Set Cover  problem parameterized by the number of sets $k$ to $DMT(IMT)$ parameterized by the budget $k$. In fact, for a Set Cover problem we can create an instance $\DMT (\IMT) = \left(G=(V, E, \omega),\right.$ $\left.A, \TP, I, t, k\right)$ as in the $\DNR(\INR)$ reduction, but for each node $x$ from the target set $T$ in the corresponding $\DNR$ problem we add a vertex $x'$, and we set $TP = \{ \{x,x'\}: x \in T\}$. In this case since all of the new nodes are disconnected from the graph, the only way to break the connections between the $\{x,x'\}$ links below the given threshold $t$ is to lower the \emph{goodness} value of the nodes $x \in T$ below the given threshold. Again, the reduction runs  polynomial time, the budget of the $DMT(IMT)$ problem is $k$, this reduction is also a parameterized reduction~\citep{parameterized}.
\end{proof}

\section{Manipulating a node directly}

\direct*
\begin{proof}
[Proof of Theorem~\ref{thm:direct}] We provide a successful strategy for the attackers. A subset $S$ of size $ \lceil 2 \times g(u_T) \times indeg(u_T) \rceil$ of the nodes in $A$ creates a new edge between each of them and the attacked node with $\omega(v, u_T) = -1$. Since the attackers want to achieve $g(u_T) < 0$, and $f(v) > \frac{1}{2}$ for every $v \in A$ before the attack, then after a successful attack $f'(v) > \frac{1}{2}$ for every $v \in A$ as well. Before the attack we have $g(u_T) = \frac{\sum_{v \in in(u_T)}f(v) \times \omega(v,u_T)}{indeg(u_T)} $. We need to show that $k =2\times indeg(u_T) + 1$  edges are enough to change the $g(u_T)$ to a value $g'(u_T)$ lower than $0$. First we observe that since $g'(u_T) < 0$ after the attack, the $f'(v) \geq f(u)$ for $v \in in(u_T)$ if $\omega(v, u_T) \geq 0$, otherwise $f'(v) \leq f(u)$ for $v \in in(u_T)$ if $\omega(v, u_T) < 0$. This implies that $\sum_{v \in in(u_T)} (f(v) \times \omega(v,u_T)) \geq \sum_{v \in in(u_T)} (f'(v) \times \omega(v,u_T))$.

In conclusion, after adding $k$ edges we obtain:
\[ \frac{1}{indeg(u_T)+k} \Big[ \sum_{v \in in(u_T)} (f(v) \times \omega(v,u_T)) - \frac{k}{2}  \Big] < 0 \] if and only if
\[ k > 2 \sum_{v \in in(u_T)} (f(v) \times \omega(v,u_T) = \]

\[2 \sum_{v \in in(u_T)} (f(v) \times \omega(v,u_T) \times \frac{indeg(u_T)}{indeg(u_T)} \leq \] \[ 2 \times g(u_T) \times indeg(u_T) \]
\end{proof}

\section{Indirect Sybil Attack}
We below provide an additional proof omitted in Section Indirect Sybil Attack.
\directtwo*
\begin{proof}
[Proof of Theorem~\ref{thm:direct2}]
Note that in the Equation~\ref{assessment} in the proof of Theorem~\ref{indirect}, one could assess the value of $\Delta g(i=t)$ by selecting $indeg_{max}(i)$ such that $\Delta g(i)$ is maximized. In this case $\Delta g(i) \leq \Big| \frac{ \sum_{u \in \pred(i)} \Delta f(u)\times \omega(u,i) }{indeg_{max}(i)} \Big| + \frac{1}{indeg_{max}(i)}$. The metric $||M||_{\infty}$ is still bounded by $\frac{1}{2}$ as in the Equation~\ref{inftymetric} which implies the result.
\end{proof}

\section{Datasets' basic statistics}
\label{sec:stats}
Figure~\ref{fig:indeg} presents the histogram of the nodes sorted per indegree. Table~\ref{tab:nums} shows how many random samples were used to simulate direct, indirect attacks for all sizes of the attacking sets considered. For mixed attacks, for each $k1,k2 \in \{1,\ldots,6\}$, $\geq 26$ samples were used for Bitcoin OTC, $\geq 12$ samples were used for Bitcoin Alpha, and $\geq 17$ samples were used for RFA Net.
\begin{figure*}[t]
    \centering
    \includegraphics[width=0.33\textwidth]{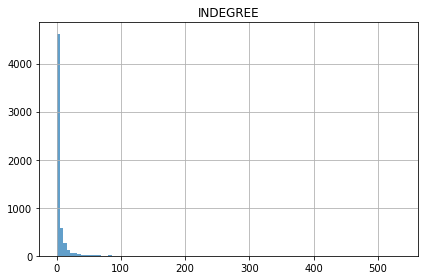}
    \includegraphics[width=0.33\textwidth]{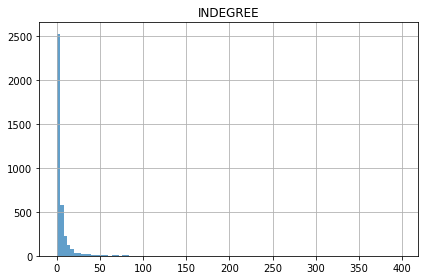}
    \includegraphics[width=0.33\textwidth]{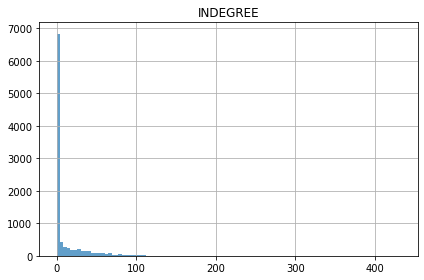}
    \caption{Histograms of indegree of the nodes of the Bitcoin OTC, Bitcoin Alpha, RFA Net networks.}
    \label{fig:indeg}
\end{figure*}

\begin{table}[H]
\begin{tabular}{llll} \hline
k & Bitcoin OTC & Bitcoin Alpha & RFA Net \\
\\ \hline
$1$ & $24$ &  $24$ & $25$  \\ 
$2$ & $21$ &  $24$ & $25$  \\
$3$ & $26$  &  $24$ & $25$  \\  
$4$ & $25$  &  $24$ & $25$  \\  
$5$ & $24$  &  $24$ & $25$  \\ 
$6$ & $23$  &  $24$ & $25$  \\  
$7$ & $22$  &  $24$ & $25$  \\ 
\hline                                   
\end{tabular}
\caption{Number of samples used to simulate direct/indirect established/not established attacks.}
\label{tab:nums}
\end{table}


\begin{center}
\begin{table}[H]
\begin{tabular}{cccc} \hline
Statistic & Bitcoin OTC & Bitcoin Alpha & RFA Net \\
\\ \hline
Size & $5881$ & $3783$ & $9654$ \\
Edges & $35592$ & $24186$ & $104554$ \\
Positive edges & $89,90\%$ & $93,64\%$ & $84\%$\\
Small in-degree $<10\%$ & $87,40\%$ & $85,70\%$ & $76\%$\\
Fair nodes $\geq 0.95$ & $99,45\%$ & $99,65\%$ & $99,99\%$ \\
Fair nodes $\geq 0.7$ & $100\%$ & $100\%$ & $100\%$\\
Goodness score $\geq 0$ & $85.85\%$ & $92.36\%$ & $93\%$\\
Goodness score $\geq 0.5$ & $2.2\%$ & $3.3\%$ & $67\%$\\
Goodness score $\leq -0.3$  & $8.2\%$ & $3.8\%$ & $0.2\%$\\

\hline
\end{tabular}
\caption{Statistics of the networks used for simulations.}
\label{tab:stats}
\end{table}
\end{center}
\section{Code}
The code package that allows running simulations presented in this paper is available under this link \url{https://github.com/irtomek/WeightPredictionsCode}.

\begin{figure*}[t]
    \centering
    \text{Bitcoin OTC} \\
    \includegraphics[width=0.33\textwidth]{images/d1.png}
    \includegraphics[width=0.33\textwidth]{images/d2.png}
    \includegraphics[width=0.33\textwidth]{images/ds.png} \\
    

    \text{Bitcoin Alpha} \\
    \includegraphics[width=0.33\textwidth]{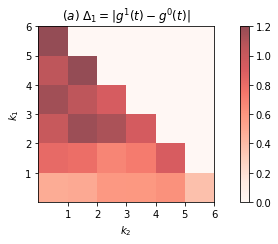}
    \includegraphics[width=0.33\textwidth]{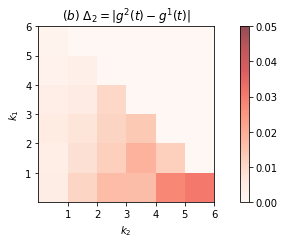}
    \includegraphics[width=0.33\textwidth]{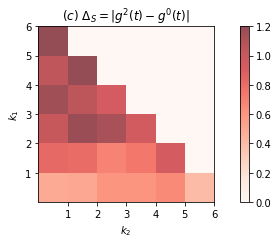} \\
    
    \text{RFA Net} \\

    \includegraphics[width=0.33\textwidth]{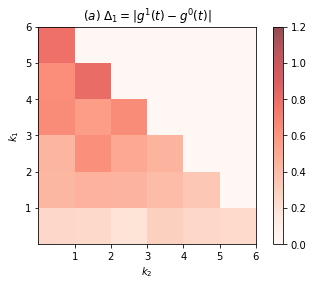}
    \includegraphics[width=0.33\textwidth]{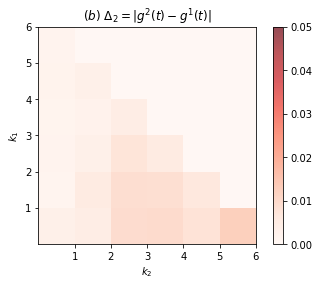}
    \includegraphics[width=0.33\textwidth]{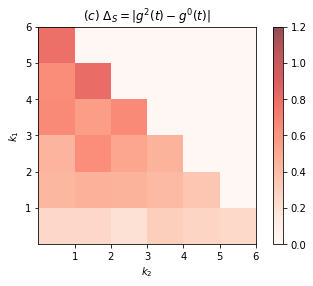}
    \caption{The results for the mixed settings in Bitcoin OTC, Bitcoin Alpha, RFA Net. The average strength of an indirect attack is small and significantly smaller that the average strength of a direct attack. $\Delta_1$ shows the influence of the attack with $k_1$ direct edges, $\Delta_2$ shows the influence of the attack with $k_2$ indirect edges, $\Delta_s$ shows the influence of the attack with $k_1$ direct edges and $k_2$ indirect edges. }
    \label{fig:deltaa}
\end{figure*}

\end{document}